\documentclass[12pt]{article}

\usepackage[utf8]{inputenc}
\usepackage{graphicx}
\graphicspath{{images/}}
\usepackage{amsmath}
\usepackage{amssymb}
\usepackage{amsthm}
\usepackage{csquotes}
\usepackage{amssymb}
\usepackage{amsfonts}
\usepackage{color}
\usepackage{xcolor}
\usepackage{comment}
\usepackage[linesnumbered,ruled,vlined]{algorithm2e}

\usepackage[margin=1in]{geometry}

\newtheorem{definition}{Definition}
\newtheorem{theorem}{Theorem}
\newtheorem{lemma}{Lemma}

\newtheorem{corollary}{Corollary}

\newtheorem{proposition}{Proposition}

\newcommand{\agents}{\mathcal{N}} 
\newcommand{\items}{\mathcal{M}} 
\newcommand{\pickers}{\Pi} 

\usepackage[]{color-edits}
\addauthor{UF}{red}

\addauthor{XH}{blue}

\begin{document}

\title{On picking sequences for chores}
\author{Uriel Feige\thanks{Weizmann Institute, Israel. {\tt uriel.feige@weizmann.ac.il}} ~and Xin Huang\thanks{Technion, Israel. {\tt xinhuang@campus.technion.ac.il}}}
\maketitle

\begin{abstract}
    We consider the problem of allocating $m$ indivisible chores to $n$ agents with additive disvaluation (cost) functions. It is easy to show that there are picking sequences that give every agent (that uses the greedy picking strategy) a bundle of chores of disvalue at most twice her share value (maximin share, MMS, for agents of equal entitlement, and anyprice share, APS, for agents of arbitrary entitlement). Aziz, Li and Wu (2022) designed picking sequences that improve this ratio to $\frac{5}{3}$ for the case of equal entitlement. We design picking sequences that improve the ratio to~1.733 for the case of arbitrary entitlement, and to $\frac{8}{5}$ for the case of equal entitlement. (In fact, computer assisted analysis suggests that the ratio is smaller than  $1.543$ in the equal entitlement case.) We also prove a lower bound of $\frac{3}{2}$ on the obtainable ratio when $n$ is sufficiently large.
    
    Additional contributions of our work include improved guarantees in the equal entitlement case when $n$ is small; introduction of the {\em chore share} as a convenient proxy to other share notions for chores; introduction of ex-ante notions of envy for risk averse agents; enhancements to our picking sequences that eliminate such envy;  showing that a known allocation algorithm (not based on picking sequences) for the equal entitlement case gives each agent a bundle of disvalue at most $\frac{4n-1}{3n}$ times her APS (previously, this ratio was shown for this algorithm with respect to the easier benchmark of the MMS).  
\end{abstract}

\section{Introduction}

We consider a setting of allocation of a set $\items$ on $m$ indivisible items to $n$ agents. An allocation $A = (A_1, \ldots, A_n)$ is a partition of $\items$ into $n$ disjoint bundles of items (some of the bundles might be empty), with the interpretation that for every $i$, agent $i$ gets bundle $A_i$. Every agent $i$ has a valuation function $v_i : 2^{\items} \rightarrow R$, assigning values to bundles of items. The utility of an allocation $A = (A_1, \ldots, A_n)$  for agent $i$ is $v_i(A_i)$.  Our allocation setting does not involve money (agents do not pay for the items, and cannot transfer utilities among themselves by paying each other). Consequently, we wish our allocation to satisfy some fairness properties. The fairness notions that we consider are discussed in Section~\ref{sec:fairness}. Here we just alert the reader that in defining these fairness notions, we shall consider both the common special case in which agents have equal entitlement to the items, and the general case of arbitrary (possibly unequal) entitlements. The entitlement of agent $i$ is denoted by $b_i$, and entitlements satisfy $b_i \ge 0$ and $\sum_{i=1}^n b_i = 1$. In the special case of equal entitlement, $b_i = \frac{1}{n}$ for every agent $i$.

In this paper we shall assume that all valuation functions are additive, meaning that for every $S \subset \items$, $v_i(S) = \sum_{e\in S} v_i(\{e\})$. Consequently, for every agent $i$ and item $e$, we can classify the item as a {\em good} if $v_i(e) \ge 0$, or as a chore $v_i(e) \le 0$ (if $v_i(e) = 0$ we may classify the item either way). We shall further assume in this paper that the allocation instance is either an instance in which all items are goods for all agents, or all items are chores for all agents. (Allocation instances in which this assumption does not hold are referred to as mixed manna, see~\cite{livanos2022almost}.) In settings of allocation of chores it will be convenient to think of items as having positive {\em disvalue} instead of negative value. We shall denote disvaluation functions by $c_i$ (where $c$ stands for {\em cost}), to distinguish them from the corresponding valuation function $v_i$ (where $v_i = - c_i$). In the context of chores we sometimes use the term {\em responsibility} instead of {\em entitlement}, so that agents with higher responsibility are expected to take upon themselves more of the chores. The notation $b_i$ will be used both for entitlements and for responsibility.

The focus of our work will be a special class of allocation mechanisms, referred to as {\em picking sequences}. To design a picking sequence for an allocation instance with a set $\agents$ of $n$ agents and a set $\items$ of $m$ items, one only needs to know the vector $(b_1, \ldots, b_n)$ of entitlements for the agents, but not the valuations $v_i$. One introduces a set $\pickers$ of identifiers for pickers, with $|\pickers| = |\agents| = n$. A picking sequence is then a vector $\pi = (\pi_1, \ldots, \pi_m)\in \pickers^m$. To use the picking sequence, there are two stages. In the preliminary stage, one chooses a bijection $f$ between agents $\agents$ and pickers $\pickers$. If both sets are $\{1, \ldots, n\}$ (which we denote by $[n]$), the bijection will often simply be the identity bijection, $f(i) = i$. However, we shall sometimes employ a more elaborate preliminary stage (see Section~\ref{sec:bijection}). Thereafter, in the main stage of using the picking sequences, there are $m$ rounds. In each round $r$, the agent $f^{-1}(\pi_r)$ whose identity is mapped to the identity of the picker $\pi_r$ gets to pick whichever single item she desires among those items not yet picked in previous rounds.

Picking sequences are used in practice. In combination with an appropriate preliminary stage, they are often perceived as fair. For example, they are used in order to allocate housing units to eligible individuals in the ``Mechir Lamishtaken" housing initiative in Israel (see website https://www.dira.moch.gov.il/, in Hebrew). There, the setting is largely that of equal entitlement, and the preliminary stage involves choosing a random permutation as the bijection between $\agents$ and $\pickers$. (This permutation is slightly modified later, to account for some aspects of unequal entitlement, but they are not of significant relevance to the current paper.) As another example, picking sequences are used by NBA teams to draft eligible basketball players (see https://www.nba.com/draft/2022). In this setting, agents (the NBA teams) have unequal entitlement (which depends on their performance in the past season), and the bijection between $\agents$ and $\pickers$ is partly deterministic (based on the entitlements), and partly random (the mapping to first four pickers is determined by a lottery, biased towards agents with higher entitlement). 

Picking sequences have several advantages that make them attractive allocation mechanisms.  The reporting burden that they enforce on agents is rather light: agents simply need to report those items that they pick when it is their turn to pick, and do not need to report their full valuation function (and not even the value that they assign to those items that they pick). Moreover, the computational burden for deciding which items to choose is often relatively small. In unit demand setting (such as the housing example above) in which an agent picks only once, all the agent needs to compute is which is the most preferred single item among the remaining items, and choosing this item is a dominant strategy for the agent (in our model, in which we assume no externalities and no collusion). In settings with additive valuations (as in this paper) and multiple picks per agent, a risk averse agent (who fears that other agents will pick those items that the agent values most) has a simple optimal strategy (in a max-min sense, maximizing the value of the worst possible resulting received bundle): in every round in which it is the agents's turn to pick, pick the item of highest value among those remaining. We refer to this strategy as the {\em greedy} picking strategy. Another advantage of picking sequences is that they are easy to understand and transparent: an agent that receives a bundle understands how it came about that this is the particular bundle that she received.

{Though picking sequences are attractive allocation mechanisms, they are not optimal in terms of the guarantees that they offer.}
Even in the special case of unit demand valuations, where the use of random order picking sequences (also referred to as random serial dictatorship, RSD) is very common, there are allocation mechanisms such as the {\em Eating mechanism} of~\cite{BM01} that stochastically dominate (in terms of utility) RSD, and in fact, stochastically dominate every distribution over picking sequences. For additive valuations over goods, no picking sequence guarantees that every agent will get at least a constant fraction of her maximin share (MMS, see Definition in Section~\ref{sec:fairness}). As one of our major goals in this paper is to offer constant approximation for share notions such as the MMS and the APS (the anyprice share, see Definition in Section~\ref{sec:fairness}), most of the current paper will not be concerned with allocation of goods, but rather with allocation of chores. For chores, it is quite straightforward to design picking sequences that provide a factor~2 approximation to standard share notions (the MMS for equal entitlement, the APS for arbitrary entitlement). Moreover, as shown in~\cite{aziz2022approximate}, one can obtain approximation ratios better than~2 (specifically, $\frac{5}{3}$) in the case of equal entitlement. This last paper serves as an inspiration for the current work. Our main technical contributions will be the design of picking sequences for chores that improve over the ratio of~2 (compared to APS) for the arbitrary entitlement case, and improve over the ratio of $\frac{5}{3}$ (compared to MMS) in the equal entitlement case.

\subsection{Fairness notions}
\label{sec:fairness}

We shall distinguish between ex-post fairness notions, and ex-ante fairness notions. Ex-post fairness notions refer to properties of the final allocation, regardless of how the allocation was obtained. Ex-ante fairness notions refer to the whole range of allocations that the allocation mechanism might produce (not only the one actually produced), where the range of allocations may depend on randomness of the allocation mechanism, and various (possibly strategic) choices made by the agents when reporting their preferences to the allocation mechanism. 

We shall mostly be interested in outcome based fairness, in which fairness is gauged by the allocation produced (or in ex-ante fairness notions, by the possible allocations that might be produced). We shall also briefly consider procedural fairness, in which fairness is gauged by the procedure by which the allocation is produced, rather than by the outcome of this procedure. 

No single fairness notion is universally agreed upon as being the correct definition of fairness. So instead, we shall define properties that are associated with fairness, with emphasize on those properties that are used in this paper.

\subsubsection{Procedural fairness}

We find it difficult to propose an exact formal definition of procedural fairness, especially in settings of unequal entitlement and arbitrary valuation functions. However, informally, we say that an allocation mechanism satisfies {\em procedural fairness} if the following properties hold:

\begin{itemize}
    \item For every two agents with the same entitlement, the allocation procedure treats them in an identical manner.
    \item For every two agents of unequal entitlement (responsibility, respectively), when allocating goods (chores, resp.), the allocation procedure treats the agent with higher entitlement (lower responsibility, resp.) more favorably than it treats the other agent. 
\end{itemize}

Let us illustrate what this procedural fairness means (and does not mean) in the context of picking sequences. For the case of agents with equal entitlement, one way of satisfying procedural fairness is by having a preliminary stage in which $\agents$ is mapped to $\pickers$ by a bijection chosen uniformly at random. With such a preliminary phase, the allocation procedure does not discriminate between agents of equal entitlements. For the case of allocating goods (chores, respectively) to agents of unequal entitlements (responsibility, resp.) and additive valuations, one way of satisfying procedural fairness is to ensure that in every prefix (suffix, resp.) of the picking sequence, an agent of higher entitlement (responsibility, resp.) gets to pick in at least as many rounds as an agent of lower entitlement (responsibility, resp.). 

\subsubsection{Share based fairness}

Share based fairness notions translate the entitlement of the agent and her valuation function into a value, and this value serves as a constraint to the allocation mechanism. Formally, a {\em share} $s$ is a function that maps a pair $(v_i,b_i)$ to a real value. An allocation is considered {\em acceptable} (with respect to the corresponding share notion) if it gives each agent $i$ a bundle $A_i$ satisfying $v_i(A_i) \ge s(v_i,b_i)$. In share based fairness notions, the agent is concerned only with the bundle that she herself receives, and is not concerned with how the remaining items are partitioned among the remaining agents. Note that the share only determines which allocations are acceptable. If there are several different acceptable allocations, one still needs a rule for selecting an allocation among the acceptable ones (or a distribution over acceptable allocations). This selection rule may be guided by the share notion (e.g., attempting to give agents a high as possible multiple of their shares), but may also be guided by other principles (e.g., maximizing welfare). In this section we mostly present various share notions, and do not address selection rules.

The {\em maximin share} (MMS) was formally defined in~\cite{{budish2011combinatorial}}. It applies only to setting with equal entitlement (namely, $b_i = \frac{1}{n}$). 

\begin{definition}
\label{def:MMS}
The MMS of agent $i$ is the minimum value of a bundle according to $v_i$, if agent~$i$ were to partition $\items$ into $n$ bundles so as to maximize this minimum. Formally, 
$$MMS(v_i,\frac{1}{n}) = \max_{A = (A_1, \ldots, A_n)} \min_{j \in [n]}[v_i(A_j)]$$ 
where $A$ ranges over all $n$-partitions of $\items$.
\end{definition}

The {\em anyprice share} (APS) was introduced in~\cite{babaioff2021fair}, and applies to setting with arbitrary entitlements. Let ${\cal{P}}_m$ denote the family of all vectors $P = (p_1, \ldots, p_m)$ of {\em prices} for the items, where $p_j \ge 0$ for every $j \in [m]$, and $\sum_{j=1}^m p_j = 1$. 
We view the entitlement of the agent as a budget $b_i$. 

\begin{definition}
\label{def:APS}
For goods (and an agent with valuation function $v_i$ and entitlement $b_i$), the APS is the highest value that the agent is guaranteed to be able to afford to buy, no matter how the goods are priced. Formally, the APS for goods is:
$$APS(v_i, b_i) = \min_{P \in {\cal{P}}} \max_{\{S\subset \items \; \mid \; \sum_{j \in S} p_j \le b_i\}} v_i(S).$$ 

For chores (and an agent with disvaluation function $c_i$ and responsibility $b_i$), the APS is the lowest disvalue that the agent is guaranteed to be able to spend her entire budget on, no matter how the chores are priced. Formally, $$APS(c_i, b_i) = \max_{P \in {\cal{P}}} \min_{\{S\subset \items \; \mid \; \sum_{j \in S} p_j \ge b_i\}} c_i(S).$$
\end{definition}

In this paper we shall introduce also another notion of a share, that we refer to as the {\em chore share} (CS). This share is applicable only for additive disvaluation functions over chores, which is the main class of valuation functions that we consider in this paper. For this class, the CS serves as a convenient replacement for the APS, and its value is never larger than that of the APS. See Definition~\ref{def:CS}. We remark that for chores, the following inequalities hold (the first inequality is applicable if and only if $\frac{1}{b_i}$ is an integer, as otherwise the MMS is not defined), and each inequality is sometimes strict:

$$MMS(c_i,b_i) \ge APS(c_i,b_i) \ge CS(c_i,b_i)$$

None of the share notions introduced above (MMS, APS, CS) is feasible: there are allocation instances with additive valuations over chores in which no allocation gives every agent her MMS~\cite{aziz2017algorithms, feige2021tight}. Hence the share notions that we shall consider are approximate versions of the above shares. For goods and $\alpha \le 1$, the $\alpha$-MMS and $\alpha$-APS are shares whose value equals $\alpha$ times the corresponding share value. For chores and $\beta \ge 1$, the $\beta$-MMS, $\beta$-APS and $\beta$-CS are shares whose value equals $\beta$ times the corresponding share value. In the current work, we design picking sequences for chores that give each agent no more than her $\beta$-CS, 
for some value of $\beta < 2$.



\subsubsection{Envy based fairness}

In envy based fairness notions, an agent is concerned with the bundle that she herself receives, and also with the bundles that other agents receive (how the remaining items are partitioned into bundles among the other agents). The standard definitions of envy freeness apply most naturally to settings with equal entitlement. There, an allocation is said to be envy free if every agent weakly prefers her own bundle over each of the other bundles allocated to the other agents. For indivisible items, envy free allocations often do not exist (e.g., if there are fewer goods than agents), and relaxations of envy freeness (such as envy free up to one item~\cite{budish2011combinatorial}, and others) are sometimes considered instead. In this paper we shall not be concerned with these relaxations. Instead, we shall consider ex-ante notions of envy freeness, which are feasible in settings with indivisible items. 

The standard notion of envy freeness is not suited for settings with arbitrary entitlements, as it is natural that agents with low entitlement will envy the bundle received by agents of higher entitlement (when allocating goods). There have been attempts to extend definitions of envy freeness to settings on unequal entitlement in some quantitative manner (considering ratios of values of bundles), but we shall follow a different approach here, which is qualitative rather than quantitative. For agents of unequal entitlement, we shall only require that the agent of higher entitlement does not envy the agent of lower entitlement. Likewise, for agents of unequal responsibility (for chores), we shall only require that the agent of lower responsibility does not envy the agent of higher responsibility.

The notion of ex-ante envy freeness is more subtle than that of ex-post envy freeness, as it involves beliefs of the agents as to what allocation (or distribution of allocations) the allocation mechanism will produce.  Allocation mechanisms that are based on picking sequences are often implemented as a multi-round game among the agents (in each round, one agent picks an item), rather as a one shot game in which all agents provide their valuation functions to some central authority, and this authority declares an allocation(as in the revelation principle). In these multi-round settings, the behavior of agents may depend on how they believe other agents will behave, and hence the output allocation is not necessary determined only by the actual valuation functions of the agents, but also by their beliefs. Consequently, our ex-ante envy notion involves some modeling of the agents (beyond just modeling them as having additive valuation functions). Specifically, we model them as being risk averse, in the sense that they believe that other agents are adversarial. For the case of allocation mechanisms based on picking sequences (the case considered in this paper), this turns out to be mathematically equivalent to the assumption that all agents have the same valuation functions, and moreover, that it is common knowledge that this is the case.  We shall refer to our ex-ante envy freeness notion as EF-RA ({\em Envy Free for Risk Averse agents}).  See more details in Section~\ref{sec:envy}.

\subsection{The preliminary stage for picking sequences}
\label{sec:bijection}

Recall that in the preliminary stage we map identities of agents $\agents$ to identities of pickers $\pickers$ in the picking sequence. 
As far as we are aware of, our paper is the first to consider picking sequences for agents of arbitrary entitlements. As such, we were faced with some new conceptual questions, which in turn lead us to a more systematic study of the preliminary stage of picking sequences. We distinguish between three types of preliminary stages.

\begin{itemize}
    \item Adversarial. The choice of bijection between $\agents$ and $\pickers$ is constrained by the vector of entitlements of the agents (the picker identity that an agent may assume may depend on the entitlement of the agent), and one may think of the given bijection as if chosen by an adversary from the set of all allowable bijections. This type of preliminary stage is mainly used for analysing ex-post share guarantees of picking sequences. 
    \item Oblivious. As in the adversarial case, there is a set of allowable bijections, but now there is a probability distribution over these bijections, and the given bijection is chosen randomly according to this probability distribution. This type of preliminary stage is mainly used for achieving improved ex-ante guarantees for picking sequences (without hurting the ex-post guarantees).
    \item Strategic. Here, the picking identities $\pickers$ are thought of as items to be allocated to agents (where the allocation needs to be a bijection), and one employs some strategic allocation mechanism for choosing the bijection. Unlike the oblivious case, the valuation functions of agents may affect which identity they get. This type of preliminary stage is used in our work in order to improve over the ex-ante guarantees offered by an oblivious preliminary stage. 
\end{itemize}

\subsection{Our results}

Some of the contributions of our work are conceptual. One such contribution is the introduction of the notion of a {\em chore share}, which serves as a convenient proxy for other share notions (MMS and APS) when agents have additive disvaluations over chores. Another such contribution is the introduction of a strategic preliminary stage for picking sequences, and illustrating its use for augmenting picking sequences that have strong ex-post share based guarantees so that they also have ex-ante envy-freeness guarantees.


For agents with arbitrary entitlements, our main result concerning ex-post guarantees is the following.

\begin{theorem}
\label{thm:arbitrary}
When allocating indivisible chores to agents with additive valuations and arbitrary entitlements, for every vector of entitlements there is a picking sequence (that can be computed in polynomial time) in which every agent (that follows the greedy picking strategy) gets a bundle of disvalue at most $1.733$ times her chore share (and hence disvalue at most $1.733$ times her APS).
\end{theorem}

We remark that Theorem~\ref{thm:arbitrary} is the first two show that when allocating indivisible chores to agents with additive valuations and arbitrary entitlements, there is an allocation that gives every agent a bundle of value at most $\rho$ times her APS, for some $\rho < 2$. (For allocation of goods, a value of at least $0.6$ times the APS was shown in~\cite{babaioff2021fair}.)  

For agents with equal entitlements, our main result concerning ex-post guarantees is the following improvement over a previous result of~\cite{aziz2022approximate}, who designed picking sequences that guarantee a $\frac{5}{3}$ approximation to the MMS.

\begin{theorem}
\label{thm:equal}
When allocating indivisible chores to agents with additive valuations and equal entitlements, there is a picking sequence (that can be computed in polynomial time) in which every agent (that follows the greedy picking strategy) gets a bundle of disvalue at most $1.6$ times her chore share (and hence disvalue at most $1.6$ times her MMS).
\end{theorem}

We remark that the $1.6$ approximation ratio to the chore share in Theorem~\ref{thm:equal} is not best possible. For small $n$ we show better bounds, and in particular, an upper bound of $\frac{13}{9}$ for $n=4$ (and a lower bound of $\frac{10}{7}$).  For general $n$, computer assisted analysis that we performed suggests that the true ratio of our picking sequences is better than $1.543$ (for the bound of $1.6$ claimed in the theorem we present a full proof verifiable by hand). The following Theorem presents a lower bound on the best possible approximation ratio achieved by picking sequences.

\begin{theorem}
\label{thm:lowerbound}
When allocating indivisible chores to $n$ agents with equal entitlements, for sufficiently large $n$ and $m$, for every picking sequence and agents that follow the greedy picking strategy, there are input instances under which some agent gets a bundle of disvalue at least $1.5$ times her MMS (and hence, also disvalue  at least $1.5$ times her APS, and  disvalue  at least $1.5$ times her chore share).
\end{theorem}

Being based on picking sequences, our results concerning ex-post guarantees easily extend to {\em best of both worlds} (BoBW) type results that offer both ex-post guarantees and ex-ante guarantees. This is done by including an appropriate preliminary stage. For ex-ante envy freeness guarantees, we use the notion of EF-RA (envy free for risk averse agents), defined in Section~\ref{sec:envy}.

For arbitrary entitlement, using a strategic preliminary stage, we get the following corollary.

\begin{corollary}
\label{cor:arbitrary}
When allocating indivisible chores to agents with additive valuations and arbitrary entitlements, for every vector of entitlements there is an allocation mechanism with the following properties:

\begin{itemize}
    \item It is ex-ante EF-RA (envy free for risk averse agents).
    \item Every agent (that follows the greedy picking strategy) gets ex-post a bundle of disvalue at most $1.733$ times her chore share.
\end{itemize}
 
\end{corollary}

For equal entitlements, using an oblivious preliminary stage (a random bijection between $\agents$ and $\pickers$), we obtain a stronger corollary that also includes ex-ante share based guarantees.

\begin{corollary}
\label{cor:equal}
When allocating indivisible chores to agents with additive valuations and equal entitlements, for every vector of entitlements there is an allocation mechanism with the following properies:

\begin{itemize}
    \item It is ex-ante EF-RA.
        \item Every agent (that follows the greedy picking strategy) gets ex-ante a bundle of expected disvalue at most her proportional share $\frac{1}{n} c_i(\items)$.
    \item Every agent (that follows the greedy picking strategy) gets ex-post a bundle of disvalue at most $\frac{8}{5}$ times her chore share. (As noted above, computer assisted analysis suggests that the ratio is in fact at most $1.543$.)

\end{itemize}
\end{corollary}

Previously, a similar corollary was stated in~\cite{BEFBoBW}, but with a weaker ex-post ratio of~$\frac{5}{3}$ instead of $\frac{8}{5}$ (and without referring to EF-RA, as that notion was not defined at the time).

For agents with equal entitlement (and additive valuations over chores), we show that allocation mechanisms not based on picking sequences provide better approximations with respect to the APS than those achievable by picking sequences (for which we present a lower bound in Theorem~\ref{thm:lowerbound}). We consider algorithm {\em AlgChores} of~\cite{BarmanK20}, for which it was shown in~\cite{BarmanK20} that it gives every agent a bundle of disvalue not larger than $\frac{4n-1}{3n}$ times her MMS. We extend this result to the more demanding benchmark of the APS.

\begin{theorem}
\label{thm:APS}
When $n$ agents have additive valuations over chores and equal entitlements, the AlgChores allocation algorithm gives every agent $i$ a bundle of cost at most $\frac{4n-1}{3n}APS_i$.
\end{theorem}

\subsection{Related work}

The maximin share (MMS) was defined by Budish~\cite{budish2011combinatorial} in studying allocation problems to agents with equal entitlements.  Kurokawa, Procaccia and Wang~\cite{procaccia2014fair} showed that there are allocation instances in which agents have additive valuations, yet no MMS allocation exists. They initiated the study of allocations that give each agent at least a $\rho$ fraction of her MMS, for $\rho$ as close to~1 as possible. Following several works~\cite{barman2018finding, GargMT19, amanatidis2015approximation, ghodsi2018fair, garg2019improved, feige2021tight}, the highest value of $\rho$ known for additive valuations over goods is $\frac{3}{4} + \frac{1}{12n}$, whereas there are examples in which a value higher than $\frac{39}{40}$ is not attainable.  For the case of chores, following several works~\cite{aziz2017algorithms, BarmanK20, huang2021algorithmic, feige2021tight},  the lowest value of $\rho$ known for additive valuations is $\frac{11}{9}$, whereas there are examples in which a value lower than $\frac{44}{43}$ is not attainable.

The Anyprice share (APS) was defined by Babaioff, Ezra and Feige~\cite{babaioff2021fair}, and is applicable also to allocation problems in which agents have unequal entitlements. For additive valuations over goods they gave an allocation algorithm that gives every agent at least $\frac{3}{5}$ of her APS, whereas for additive valuations over chores they observed that known techniques provide allocations in which every agent gets a bundle of disvalue at most twice her APS.

In our work we study picking sequences for agents with additive valuations over chores, and derive approximation guranteess for them with respect to the MMS (for the case of equal entitlements) and the APS (for the case of arbitrary entitlements). Our study was inspired by the work of Aziz, Li and Wu~\cite{aziz2022approximate} who designed picking sequences for the equal entitlements case that give each agent a bundle of disvalue at most $\frac{5}{3}$ times her MMS. In our work, we design new picking sequences, improving the approximation ratio in the equal entitlements case, and getting the first nontrivial ratios (better than~2) in the arbitrary entitlement case. As a side result, unrelated to picking sequences, we analyze the performance of the allocation algorithm {\em AlgChores} of~\cite{BarmanK20}, and show that the approximation ratio of $\frac{4n-1}{3n}$ previously proved compared to the MMS, also holds with respect to the APS (in the equal entitlement setting). 


Additional themes in the study of fair allocation of indivisible items can be found (for example) in the survey~\cite{aziz2022algorithmic}. We briefly mention some of these themes. There are fairness notions not based on shares, and among them those based on variations on envyfreeness (such as EF1 and EFX) receive much attention (see for example~\cite{caragiannis2016unreasonable,DBLP:conf/sigecom/ChaudhuryGM20,chaudhury2021improving,berger2022almost}). 
Besides setting with either only goods or only chores, there are studies of fair allocation of {\em mixed manna} that contain both goods and chores~\cite{aziz2018fair,livanos2022almost}. 
There are also studies of fair allocation in settings in which the valuation functions of agents go beyond additive. See for example~\cite{ghodsi2022fair,chaudhury2022fair,li2022fair}. 

\subsection{Overview of our proof techniques}

Our main technical results involve picking sequences for chores for agents with additive disvaluation functions. In this section we provide an overview for the main ideas used in our proofs for these results.

\subsubsection{The chore share}
\label{sec:CS}

Rather than analyse the approximation ratios of our picking sequences directly compared to the APS, we introduce the {\em chore share} (CS), which is a lower bound on the APS (for additive disvaluation functions over chores). Hence, an approximation of $\rho$ with respect to the CS is also an approximation ratio of $\rho$ with respect to the APS (recall, that agents wish to receive a bundle of small disvalue). The advantage of using the CS as a proxy to the APS is that it simplifies the analysis, compared to direct use of the APS. 

For the purpose of defining $CS_i$, the chore share of agent $i$ that has additive disvaluation function $c_i$ over chores, we assume that chores are ordered from highest disvalue (chore $e_1$) to lowest disvalue (chore $e_m$) according to $c_i$.

\begin{definition}
\label{def:CS}
The {\em chore share} of agent $i$ with disvaluation function $c_i$ and responsibility (entitlement) $b_i$, denoted by $CS_i$, is the minimum value of $z$ satisfying the following three inequalities:

\begin{enumerate}
    \item $z \ge b_i \cdot v_i(\items)$. (The chore share is at least the proportional share.)
    \item $z \ge v_i(e_1)$. (The chore share is at least the disvalue of the item with highest disvalue.)
    \item $z \ge v_i(e_k) + v_i(e_{k+1})$, where $k = \lfloor \frac{1}{b_i} \rfloor$. (For example, if $b_i = 0.3$, then $z \ge v_i(e_3) + v_i(e_4)$. Note that $APS_i \ge v_i(e_3) + v_i(e_4)$, by pricing each of the first four items at $0.25$.)
\end{enumerate}

\end{definition}

As an example of a gap between the $CS_i$ and $APS_i$, consider an agent $i$ of entitlement $\frac{1}{n}$, and $2n + 1$ chores, each of value $\frac{n}{2n+1}$. In this example $CS_i = 1$, whereas $APS_i = MMS_i = \frac{3n}{2n+1}$. 

\subsubsection{Identically ordered (IDO) instances}
\label{sec:IDO}

{When agents have additive disvaluations (and likewise, additive valuations), {\em identical ordering} (IDO) input instances are of special interest. 

\begin{definition}
\label{def:IDO} 
We say that an input instance with additive disvaluation functions over indivisible chores has {\em identical ordering} (IDO) if the disvaluation functions of agents are such that for every disvaluation function $c_i$ and every two items $e_j$ and $e_k$ with $j < k$, it holds that $c_i(e_j) \ge c_i(e_k)$.
\end{definition} 

For picking sequences for chores when agents have additive disvaluation functions, we  consider the {\em greedy picking strategy}: at every round in which it is the agent's round to pick, she picks the item of smallest disvalue among those remaining.

\begin{lemma}
\label{lem:IDO}
For every picking sequence and every disvaluation function of an agent, the worst case for an agent who uses the greedy picking strategy (in terms of the disvalue of the final bundle received) is when the input instance is IDO and all other agents also use the greedy picking strategy.
\end{lemma}

\begin{proof}
For IDO instances, when all agents use the greedy picking strategy, in every round $r$ in which it is the agent's turn to pick, the selection of chores remaining to pick from is the worst possible (the $m-r+1$ chores of highest disvalue). 
\end{proof}

We analyse the approximation guarantee compared to the chore share for agents that follow the greedy picking strategy. By Lemma~\ref{lem:IDO}, we may assume for this purpose that the input instance is an IDO instance, and all agents follow the greedy picking strategy. (By a well known lemma from~\cite{bouveret2016characterizing}, this assumption can be made also when the allocation mechanism is not based on a picking sequence, a fact that we use in the proof of Theorem~\ref{thm:APS}.) Under this assumption, there is a one to one correspondence between picking sequences and allocations. The agent picking in round $r$ gets item $e_{m-r+1}$ in the allocation, and conversely, the agent that gets item $i$ is the one who picks in round $m-i+1$. This leads to the following corollary.

\begin{corollary}
\label{cor:IDO}
Consider an arbitrary allocation instance with $n$ agents and $m$ chores $e_1, \ldots e_m$. Consider an arbitrary allocation $A = (A_1, \ldots, A_n)$, and associate with it a picking sequence $\pi$ in which every agent $i$ picks in those rounds $r$ for which $e_{m - r + 1} \in A_i$. Let $\rho_i(A)$ be the worst possible approximation ratio $\frac{c_i(A_i)}{APS_i}$ over all additive disvaluation functions for agent $i$ with respect to which the input instance is an IDO instance ($c_i(e_j) \ge c_i(e_k)$ for every $j < k$). Let $\rho_i(\pi)$ be the worst possible approximation ratio (of the disvalue of the bundle received compared to the APS) for agent $i$ that follows the greedy picking strategy in picking sequence $\pi$, taken over all additive disvaluation functions $c_i$, and over all picking strategies for the other agents. Then $\rho_i(A) = \rho_i(\pi)$. Moreover, the same equality holds with $APS_i$ replaced by $CS_i$, and in the case of equal entitlement, also with $APS_i$ replaced by $MMS_i$.
\end{corollary}

With Corollary~\ref{cor:IDO} in mind, when describing our picking sequences, we shall sometimes describe them as allocations (e.g., agent~1 gets item $e_1$) rather than as picking sequences (e.g., agent~1 gets to pick in round $m$).} 

\subsubsection{Picking sequences for chores for agents with arbitrary entitlements} 

We now present an overview for the proof of Theorem~\ref{thm:arbitrary}.  We assume without loss of generality that the instance is IDO. Then, we consider a fractional allocation of the chores, and round this fractional allocation to get an integral allocation. {By Corollary~\ref{cor:IDO},} this integral allocation defines a picking sequence (the agent getting chore $e_j$ in the integral allocation gets to pick in round $m-j+1$ in the picking sequence). Lemma~\ref{lem:FracToSequence} (in Section~\ref{sec:arbitrary}) is a key lemma connecting between the properties of the fractional allocation and the approximation guarantees of the resulting picking sequence. Starting from the {\em proportional fractional allocation} (in which every agent $i$ gets a fraction $b_i$ of each chore), the lemma implies that after rounding, the picking sequence will approximate the chore share within a ratio no worse than~2. To improve over the ratio of~2, we do not immediately round the proportional fractional allocation, but instead modify it gradually to a different mixed allocation (partly integral, partly fractional), in which the first $n$ items are allocated integrally, and the remaining items are allocated fractionally. We show that this can be done in such a way that applying Lemma~\ref{lem:FracToSequence} to the mixed allocation gives an improved approximation ratio of 1.733.

The approach that we use for handling the arbitrary entitlement case (rounding a fractional allocation using Lemma~\ref{lem:FracToSequence}) has the advantage that the analysis does not get too complicated (despite the fact that we need to handle arbitrary vectors of entitlements $(b_1, \ldots, b_n)$). However, it does not lead to best possible approximation ratios. 

\subsubsection{Picking sequences for chores for agents with equal entitlements} 
\label{sec:equalOverview}

We present an approach for constructing picking sequences that has the potential to produce best possible approximation ratios. Analysing this approach is rather complicated, and hence we use this approach only in the simpler case of agents with equal entitlements. This approach is used in the proof of Theorem~\ref{thm:equal}, which improves over approximation ratios proved in~\cite{aziz2022approximate} for other picking sequences.

Our approach for proving Theorem~\ref{thm:equal} is as follows. As in the proof of Theorem~\ref{thm:arbitrary}, we assume without loss of generality that the input instance is IDO. Recall the correspondence between picking sequences and allocations {(Corollary~\ref{cor:IDO})}.
Our picking sequences correspond to allocations that start with what we shall refer to as a {\em ridge}. Namely, the first $n$ items are given to agents~1 to~$n$ (in increasing order), and the next $n$ items are given to agents $n$ to~1 (in decreasing order). In other words, among the first $2n$ items, every agent $i$ gets items $e_i$ and $e_{2n-i+1}$. We refer to allocations that start with a ridge as {\em ridge picking orders}. {(The distinction that we make between picking orders and picking sequences is that in picking orders it is assumed that agents select in their turn the worst possible chore instead of the best possible chore. See Section~\ref{sec:order}.)} For every fixed $n$, among all ridge picking orders (that differ from each other by the allocation of chores not in the ridge), we attempt to design the ridge picking order with best approximation ratio $\hat{r}_n$ compared to the chore share. For fixed small $n$, we determine $\hat{r}_n$ exactly, whereas for large $n$, we obtain upper bounds on $\hat{r}_n$. 

To determine the best possible value $\hat{r}_n$, we design a procedure that checks for a candidate value $r$ whether there is a ridge picking order with approximation ratio no worse than $r$ (and thus establish that $\hat{r}_n \le r$).

In our ridge picking orders, every agent $i$ is associated with a period $p_i$ for receiving items. This period need not be an integer. The period $p_i$ determines a sequence of thresholds on the indices of items that agent is allowed to receive. That is, the $t$th item that agent $i$ receives has index no smaller than the $t$th threshold in her sequence of thresholds. The period $p_i$ of each agent $i$ is chosen as the smallest possible value that ensures that the disvalue of the chores received by the agent does not exceed $r$ times her chore share.

Our choice of periods $p_i$ will need to ensure that that all items can be  allocated, without violating any of the thresholds. Equivalently, it will need to ensure that for every $j$, the number of thresholds of value at most $j$ in all lists is at least $j$. We refer to this as the {\em covering constraints}.  A necessary requirement for the covering constraint to hold (when $m$ tends to infinity) is that $\sum_{i=1}^n \frac{1}{p_i} \ge 1$. We refer to this as the {\em fractional covering constraint}. However, satisfying the fractional covering constraint is not a sufficient condition for satisfying all integer covering constraints. Moreover, there does not seem to be a simple characterization for those choices of periods $(p_1, \ldots, p_n)$ that satisfy the covering constraints for ridge orders. This is the main technical difficulty in determining the approximation ratios obtainable by ridge picking orders. 

We employ the following approach to obtain a nearly tight upper bound on $\hat{r} = \sup \hat{r}_n$. We prove in Lemma~\ref{lem:doublen} that for every $n$, $\hat{r}_n \le \hat{r}_{2n}$. Consequently, $\hat{r}_n \le \hat{r}_{2^kn}$ for every integer $k \ge 0$. Then, in Lemma~\ref{lem:8r} we prove that for every $n$ divisible by~8, $\hat{r}_n \le \frac{8}{5}$. The proof of this lemma involves partitioning the agents into~8 blocks of super agents, and designing a picking order for the super agents that can be lifted back to a picking order for the original agents, without losing in the approximation ratio. The combination of the two lemmas then implies that for every $n$,  $\hat{r}_n \le \hat{r}_{8n} \le \frac{8}{5}$.

Our approach can lead to improved bounds by considering values of $n$ divisible by higher powers of~2 (instead of by~8, as in done in Lemma~\ref{lem:8r}). However, then the picking orders become longer to describe, and verifying correctness by hand becomes tedious. Running a computer program written for the purpose of designing and verifying these picking orders shows that $\hat{r} \le 1.543$, but we have not produced a readable proof for such a claim.

\subsubsection{A lower bound}

In the proof of Theorem~\ref{thm:lowerbound}, we make use of ridge picking orders and of the fractional covering constraint, both described in Section~\ref{sec:equalOverview}. We first observe that for every $n \ge 2$, if the picking order is not a ridge picking order, then the approximation ratio compared to the MMS is no better than $\frac{3}{2}$. Hence for the purpose of proving Theorem~\ref{thm:lowerbound}, we may consider only ridge picking orders. Then, using the fractional covering constraint, we prove that for sufficiently large $n$, there is no ridge picking order with approximation ratio better than 1.52, compared to the chore share. (In contrast, for small $n$, ridge picking orders do achieve approximation ratios better than 1.5.) Finally, we note that as $m$ grows, the MMS for our negative examples approaches the chore share, and hence when $m$ is sufficiently large, the 1.52 approximation lower bound holds also with respect to the MMS.

\subsubsection{BoBW results}

In the proof of Corollary~\ref{cor:arbitrary}, we add a preliminary strategic stage so as to enhance the picking sequence designed for Theorem~\ref{thm:arbitrary} by an ex-ante EF-RA property. A risk averse agent $i$ translates a picking sequence $\pi$ and an identity $j$ in the picking sequence to a bundle $A_j$ (for every round $r$ in which $j$ picks an item according to $\pi$, the bundle $A_j$ contains chore $e_{m-r+1}$). Hence agent $i$ may think of each identity $j$ as a an ``item" of disvalue $c_i(A_j)$. To prevent ex-ante envy among risk averse agents of equal entitlement, we may simply assign identities to agents via a uniformly random bijection (and this is used in the proof of Corollary~\ref{cor:equal}). For the case of arbitrary entitlement (or responsibility, for chores), we need agents of low responsibility not to envy agents of high responsibility.  This is achieved by a strategic preliminary phase composed of a picking sequence for identities, where in this preliminary picking sequence agents of low responsibility get to pick before agents of high responsibility (and for agents of equal responsibility, the order among them in the picking sequence is determined uniformly at random). 
We may refer to this preliminary picking sequence as {\em priority random serial dictatorship} (PRSD). 

Two remarks are in order here. One is that the use of a strategic PRSD preliminary stage (or just RSD, in this case)  may be useful even if all agents have equal entitlements (in our proof of 
Corollary~\ref{cor:equal} we just use a random bijection), as its outcome for assigning identities to agents stochastically dominates that of a random bijection (if agents are risk averse and associate disvalues with identities as explained above). The other is that the use of a preliminary PRSD stage may be useful also for picking sequences for goods, not only for chores. (For goods, picking sequences do not offer good ex-post guarantees compared to the APS, but nevertheless, they might still be used in practice).

\subsection{Discussion and open problems}

In this work we consider picking sequences for chores for agents with additive valuations. We presented two techniques for designing such picking sequences. One is based on rounding of fractional allocations, and was used in the case of arbitrary entitlement (the proof of Theorem~\ref{thm:arbitrary}). The other is based on ridge sequences that satisfy certain covering contraints, and was used in the case of equal entitlement so as to get improved approximation ratios (the proof of Theorem~\ref{thm:equal}). This second technique can be adapted to the case of arbitrary entitlement, but we leave open the question of what approximation ratio it gives in that case.

More generally, we leave open the question of whether the best approximation ratios possible (compared to either the chore share, or the APS) in the case of arbitrary entitlement are the same as they are in the case of equal entitlement. This question is open for allocation algorithms in general, and for picking sequences in particular.

Another question concerns {\em best of both worlds} (BoBW) results, as in corollaries~\ref{cor:arbitrary} and~\ref{cor:equal}. Known techniques can be used in order to derive a BoBW result showing that for agents with arbitrary entitlement (and additive disvaluation functions over chores), there is an allocation mechanism that ex-ante gives each agent a bundle of expected value not larger than her proportional share, and ex-post not larger than twice the APS. Moreover, such a mechanism can be implemented as a distribution over picking sequences. (The proof of the above claims can be obtained by starting with the proportional fractional allocation, and transforming it into a distribution over picking sequences, using principles as in the proof of Lemma~\ref{lem:FracToSequence}. Further details are omitted.)
However, it is open whether there is such  BoBW result where the ex-post ratio is strictly smaller than~2.

\section{Picking sequences for arbitrary entitlements}
\label{sec:arbitrary}

In this section we consider $n$ agents with arbitrary entitlements and additive disvaluation functions over $m$ indivisible chores. The entitlement (also referred to as responsibility) of agent $i$ is denoted by $b_i$, her disvaluation function is denoted by $c_i$, and her chore share (see Definition~\ref{def:CS}) is denoted by $CS_i$. 
We assume (without loss of generality) that the instance is an IDO instance (see Definition~\ref{def:IDO}), and hence $c_i(e_j) \ge c_i(e_k)$ for every agent $i$ and every $j < k$. We also assume (without loss of generality) that agents are ordered in order of increasing entitlement ($b_1 \le b_2 \ldots \le b_n$).

In constructing our picking sequences, we shall consider fractional allocations. Let us first introduce some notation.
Fix the entitlements $(b_1, \ldots, b_n)$ of the agents. A fractional allocation $A^f$ is a collection of nonnegative coefficients  $\{a_{ij}\}_{i \in \agents, j\in \items}$, where $a_{ij}$ denotes the fraction of item $e_j$ allocated to agent $i$, and for every item $e_j$, $\sum_i a_{ij} = 1$. For each agent $i$, let $f_i$ denote the index of the first item $e_k$ for which $a_{ik} \not= \{0,1\}$ ($i$ gets a strict fraction of $e_k$). We may assume without loss of generality that $f_i$ indeed exists for every $i$ (for example, by adding to the instance a single chore $e_{m+1}$ of disvalue~0 to all agents, and setting $a_{i, m+1} = b_i$ for every agent $i$).

We use the following lemma.

\begin{lemma}
\label{lem:FracToSequence}
Consider arbitrary entitlements $(b_1, \ldots, b_n)$ for the agents and an arbitrary fractional allocation $A^f$, and recall the notation above. Fix some $\rho > 1$.
Suppose that for every $i$, for every additive disvaluation function $c_i$ (that one can associate with agent $i$) it holds that 

$$c_i(e_{f_i}) + \sum_{e_j \in {\items}} a_{ij} \cdot c_i(e_j) \le \rho \cdot CS_i$$

Then there is a picking sequence for $n$ agents and entitlements $(b_1, \ldots, b_n)$ in which every agent $i$ with an additive disvaluation function that uses the greedy picking strategy gets a bundle of disvalue at most $\rho \cdot CS_i$. Moreover, given $A^f$, such a picking sequence can be designed in polynomial time.
\end{lemma}

\begin{proof}
Given a fractional allocation $A^f$, we design a family $\cal{F}$ of picking sequences, and every member of the family will satisfy the conclusions of the lemma. Recall (from Section~\ref{sec:IDO}) that there is a one to one correspondence between picking sequences and allocations for IDO instances. We shall present our picking sequences by describing their corresponding allocations. 

For every agent $i$ and $t \ge 1$, let $n_{i}^t$ denote the number of chores received by agent $i$ among the first $t$ chores, with $n_i^0 = 0$. Then a picking sequence is in $\cal{F}$ if and only if, for every $t \ge 1$, the agent who receives chore $e_t$ is an agent $i$ for which the inequality $\sum_{j=1}^t a_{ij} > n_i^{t-1}$ holds.

The above indeed defines a picking sequence in which all items are allocated. Namely, for every $t \ge 1$, there is an agent $i$ for which the inequality $\sum_{j=1^t} a_{ij} > n_i^{t-1}$ holds. This holds by an averaging argument, observing that $\sum_{i \in \agents} n_i^{t-1} = t-1$ whereas $\sum_{i \in \agents} \sum_{j \le t} a_{ij} = t > t-1$. 

We now prove that every agent $i$ receives a bundle $A_i$ of disvalue at most $v_i(e_{f_i}) + \sum_{e_j \in {\items}} a_{ij} \cdot v_i(e_j)$ (and hence at most $\rho \cdot CS_i$). Fixing $i$, consider only nonzero variables $a_{ij}$. Let $q$ denote the number of variables $a_{ij}$ of value~1 whose $j$ index precedes $f_i$. Let $i_1, i_2, \ldots $ be the indices of chores in $A_i$ (hence,  $A_i = \{e_{i_1}, e_{i_2}, \ldots \}$). Then by the definition of $\cal{F}$ it holds that $i_{q+1} \ge f_i$. Remove chore $e_{i_{q+1}}$ from $A_i$. We can now have a fractional matching in which every item $e_{i_k}$ in $A_i \setminus \{e_{i_{q+1}}\}$ is matched with one unit $u_{i_k} = 1$, where $u_{i_k}$ is obtained as a sum of fractions of items in $\items$ whose index is not larger than $i_k$, and moreover, for every item $e_{\ell} \in \items$ the sum of fractions of $e_{\ell}$ that is matched is at most $a_{i\ell}$. This matching combined with the fact that the instance is ordered implies that the sum of disvalues of items picked by agent $i$, excluding item $e_{i_{j+1}}$, is at most $\sum_{e_j \in {\items}} a_{ij} \cdot c_i(e_j)$. Using the fact that $i_{q + 1} \ge f_i$ we have that $c_i(e_{q+1}) \le c_i(e_{f_i})$, and thus  $c_i(A_i) \le c_i(e_{f_i}) + \sum_{e_j \in {\items}} a_{ij} \cdot c_i(e_j)$.
\end{proof}

We now construct fractional allocations satisfying the condition of Lemma~\ref{lem:FracToSequence} with $\rho$ smaller than~2. 
Recall that the chore share CS has the following properties (for every agent $i$ with entitlement $b_i$):

\begin{itemize}
    \item $CS_i \le APS_i$.
    \item $c_i(e_1) \le CS_i$.
    \item $c_i(e_k) \le \frac{1}{2}CS_i$ for every item $k$ with $k > \lfloor \frac{1}{b_i} \rfloor$.
    \item $CS_i \ge b_i \cdot v_i(\items)$.
\end{itemize}

To simplify notation and terminology, and without affecting the generality of the results, we assume that $CS_i = 1$ 
(this can be obtained by scaling $c_i$ by a multiplicative factor), and that $CS_i = b_i \cdot c_i(\items)$ (this can be obtained by raising the values of some low valued items, and thus every bundle that now has value at most $\rho \cdot CS_i$ also had such a value before the raise). Hence the properties above become:

\begin{itemize}
    \item $CS_i = 1 \le APS_i$.
    \item $c_i(e_1) \le 1$.
    \item $c_i(e_k) \le \frac{1}{2}$ for every item $k$ with $k > \lfloor \frac{1}{b_i} \rfloor$.
    \item $CS_i = b_i \cdot c_i(\items)$. (The chore share equals the proportional share.)
\end{itemize}

For vectors $\bar{\alpha} = (\alpha_1, \ldots, \alpha_n)$ and $\bar{\beta} = \{\beta_1, \ldots, \beta_n\}$ we refer to a fractional allocation $A^f$ as a $(\bar{\alpha}, \bar{\beta})$ allocation if for every agent $i$ it holds that $\sum_{e_j \in {\items}} a_{ij} \cdot c_i(e_j) \le \alpha_i \cdot CS_i$ and $c_i(e_{f_i}) \le \beta_i \cdot CS_i$. Using the convention that $CS_i = 1$ this becomes $\sum_{e_j \in {\items}} a_{ij} \cdot v_i(e_j) \le \alpha_i$ and $v_i(e_{f_i}) \le \beta_i$.

Observe that Lemma~\ref{lem:FracToSequence} implies that every $(\bar{\alpha},\bar{\beta})$ fractional allocation can be transformed into a picking sequence with $\rho \le \max_i[\alpha_i + \beta_i]$. For the proportional fractional allocation (every agent $i$ gets a $b_i$ fraction of every item) we have that $\alpha_i, \beta_i \le 1$ and hence this gives a picking sequence with $\rho \le 2$ (for agents with arbitrary entitlements). We now present an algorithm (running in polynomial time) that transforms the proportional fractional allocation into a new fractional allocation with $\alpha_i + \beta_i \le 1.733$. Crucially, the algorithm only uses the entitlements $b_i$, but not the disvaluation functions $c_i$ (so that Lemma~\ref{lem:FracToSequence} can be applied on this fractional allocation). Denoting the vector $(b_1, \ldots, b_n)$ of entitlements by $\bar{b}$, we denote our fractional allocation by $A^f(\bar{b})$. 

We present the steps of constructing the fractional allocation $A^f(\bar{b})$. To clarify these steps, we shall have two running examples. In example $E1$ there are $n$ agents of equal entitlement. In example $E2$ there are $n=3$ agents with entitlements $(\frac{1}{8}, \frac{3}{8}, \frac{1}{2})$. In both cases, the number of items can be thought of as infinite (by adding items of~0 value), and the instance is ordered (earlier items have disvalue at least as high as later items, for all agents).

Preliminary step:

\begin{itemize}

     \item Sort the agents in order of increasing entitlement $b_1 \le b_2 \le \ldots \le b_n$. For every $i \le n$ we set $B_i = \sum_{j \le i} b_j$. In $E_1$, $B_i = \frac{i}{n}$ for every $1 \le i \le n$. In $E2$, $B_1 = \frac{1}{8}$, $B_2 = b_1 + b_2 = \frac{1}{2}$, and $B_3 = b_1 + b_2 + b_3 = 1$.
    
\end{itemize}

Main steps:   

\begin{enumerate}

    \item Start with the proportional fractional allocation (in which every agent $i$ gets a fraction $b_i$ of every item). We refer to this initial allocation as $A^1$. In $E1$ every agent gets a $\frac{1}{n}$ fraction of every item. In $E2$, of every item, the agents get fractions $\frac{1}{8}, \frac{3}{8}, \frac{1}{2}$, respectively. 
    
    \item Modify $A^1$ into a fractional allocation $A^2$, where $A^2$ is possibly not a legal fractional allocation, in the sense that for some items the total fraction allocated might be more than~1 (they have a {\em surplus}), and for some other items it might be less than~1 (they have a {\em deficit}). $A^2$ is the sum of two allocations, an integral allocation $A^{2I}$ and a fractional allocation $A^{2f}$. 

    \begin{enumerate}

    \item In the integral allocation $A^{2I}$, for every $1 \le i \le n$, item $e_i$ is given to agent $i$ integrally, whereas all other items are not allocated.

    \item The fractional allocation $A^{2f}$ is a modification of the proportional fractional allocation $A^1$. In this modification, every agent $i$ gives up fractions of items that total~1, starting from $e_1$ (and thus ending at item $e_j$ for $j = \lceil \frac{1}{b_i} \rceil$). 

    \end{enumerate}

    We now describe $A^{2f}$ and the resulting $A^2$ for our two running examples.

    \begin{itemize}

        \item In $E1$, allocation $A^1$ gives every agent a fraction of $\frac{1}{n}$ in every item. Consequently, in $A^{2f}$ every agent gives up her fractions in the first $n$ items. Thus $A^{2f}$ does not allocate any of the first $n$ items, and allocates each of the remaining items in a proportional way. Consequently, $A^2 = A^{2I} + A^{2f}$ is a legal fractional allocation.

        \item For $E2$, in allocation $A^{2f}$ agent~1 gives up her fraction in the first~8 items, agent~2 gives up her fraction in the first two items, and also decreases her fraction in item $e_3$ from $\frac{3}{8}$ to $\frac{1}{8}$, and agent~3 gives up her fraction in the first two items. Consequently, $A^2 = A^{2I} + A^{2f}$ is not a legal fractional allocation. In $A^2$, items $e_1$ and $e_2$ are each allocated once (in $A^{2I}$), but item $e_3$ is allocated to an extent of $1 + \frac{5}{8}$ (allocated once in $A^{2I}$, and $0 + \frac{1}{8} + \frac{1}{2} = \frac{5}{8}$ in $A^{2f}$), whereas each of items $\{e_4, \ldots, e_8\}$ is allocated only to the extent of $\frac{7}{8}$ ($0 + \frac{3}{8} + \frac{1}{2} = \frac{7}{8}$ in $A^{2f}$). Hence item $e_3$ has a surplus of $\frac{5}{8}$, and each of items $\{e_4, \ldots, e_8\}$ has a deficit of $\frac{1}{8}$. Observe that by construction, the sum of surpluses equals the sum of deficits.

    \end{itemize}
    
    \item In this step, we modify $A^2$ to a legal allocation $A^3$. The modification is done by moving fractions from items with surpluses (the fractions moved from these items are the fractions contributed by $A^{2f}$, so that the respective item remains allocated integrally according to $A^{2I}$) to items with deficits in an arbitrary way, so as to eliminate all surpluses and deficits. 
    For $E1$, this modification step is empty, because $A^2$ was legal, and then $A^3 = A^2$. For $E2$ we need to move the surplus of $0 + \frac{1}{8} + \frac{1}{2} = \frac{5}{8}$ from $e_3$ to cover the deficits of $\{e_4, \ldots, e_8\}$. For concreteness, we choose to move the fraction associated with agent~2 to item $e_4$, and to partition the fraction associated with agent~3 among the items $\{e_5, \ldots, e_8\}$. Hence allocation $A^3$ is as follows. Items $\{e_1, e_2, e_3\}$ are allocated integrally to agents~1,~2 and~3. Item $e_4$ is allocated in ratios $(\frac{1}{2}, \frac{1}{2})$ to agents~2 and~3. Each of items $\{e_5, \ldots, e_8\}$ is allocated in ratios $(\frac{3}{8}, \frac{5}{8})$ to agents~2 and~3. Each of the remaining items is allocated in ratios $(\frac{1}{8}, \frac{3}{8}, \frac{1}{2})$ to the three agents.
    
    \item In this step we modify the legal fractional allocation $A^3$ to an illegal fractional allocation $A^4$. Crucially, $A^4$ will be illegal only in the sense that items $e_j$ with $j > n$ may have a surplus, but no item will have a deficit, a fact that will be proved in Lemma~\ref{lem:noDeficit}. The modification is as follows. The allocation of the first $n$ items does not change (recall that they are allocated integrally according to $A^{2I}$). As to the remaining items $e_j$ (with $j > n$), for every agent $i$, her fractional allocation is scaled by a multiplicative factor. This multiplicative factor is a function of $B_i$ (recall the definition of $B_i$ from the preliminary step), and this function will be denoted by $s(x)$. For the purpose of illustrating this step on our examples $E_1$ and $E_2$, we shall tentatively use the function $s(x) = \frac{1}{2} + x$. Hence for $j > n$ we have $A^4_{ij} = (\frac{1}{2} + B_i)A^3_{ij}$. We alert the reader that the actual function $s(x)$ used in constructing our fractional solution is a function different than $\frac{1}{2} + x$, and will be described later. (The choice $s(x) = \frac{1}{2} + x$ can be shown to imply an approximation ratio of $\rho \le \frac{57}{32}$, and our later choice of $s(x)$ is designed so as to give an even lower value of $\rho$.)
    
    For $E1$, for every $j > n$ and agent $i$, item $e_j$ is allocated in $A^4$ to an extent of $(\frac{1}{2} + \frac{i}{n})\frac{1}{n}$ to agent $i$. Hence the item is allocated to an extent of $\sum_{i=1}^n (\frac{1}{2} + \frac{i}{n})\frac{1}{n} = \frac{1}{2} + \frac{1}{n^2} \cdot \frac{n(n+1)}{2} = 1 + \frac{1}{2n}$, and has a surplus of $\frac{1}{2n}$ (that tends to~0 as $n$ grows). 
    
    
    For $E2$ the scaling factors for the agents are $\frac{1}{2} + \frac{1}{8} = \frac{5}{8}$, $\frac{1}{2} + \frac{1}{2} = 1$, and $\frac{1}{2} + 1 = \frac{3}{2}$, respectively. Hence item $e_4$ is allocated in fractions $(\frac{1}{2}, \frac{3}{4})$ to agents~2 and~3, each of items $\{e_5, e_6, e_7, e_8\}$ is allocated in fractions $(\frac{3}{8}, \frac{15}{16})$ to agents~2 and~3, and every item $e_j$ with $j > 8$ is allocated in ratios  $(\frac{5}{64}, \frac{3}{8}, \frac{3}{4})$. Observe that indeed every item $e_j$ with $j > n = 3$ has a surplus.

    \item In this step we modify the illegal fractional allocation $A^4$ to a legal fractional allocation $A^5$, that will serve as our final fractional allocation $A^f(\bar{b})$. For every item that has a surplus, reduce in an arbitrary way the fractional allocations of that item until there is no surplus. For $E1$, this can be done by reducing the fractional allocation of agent~$n$ from $\frac{3}{2n}$ to $\frac{1}{n}$, for every item $e_j$ with $j > n$. For $E2$, this can be done by allocating every item $e_j$ with $j > n = 3$ in ratios $(\frac{3}{8}, \frac{5}{8})$ to agents~2 and~3 (and then agent~1 gets only item $e_1$). 
    
\end{enumerate}

\begin{lemma}
\label{lem:noDeficit}
Suppose that the function $s(x)$ is monotone non-decreasing for $0 \le x \le 1$, and moreover, that $\int_0^1 s(x) \; dx \ge 1$. Then in the procedure described above, the final fractional allocation $A^5$ is a legal allocation (every item is allocated exactly once).
\end{lemma}

\begin{proof}
Inspection of the procedure above shows that the only nontrivial part of the lemma is to prove that in allocation $A^4$, no item has a deficit. For this, we partition the items into three classes.

\begin{enumerate}

    \item Items $\{e_1, \ldots, e_n\}$. In $A^4$, each of these items is allocated integrally as in $A^{2I}$, and hence has no deficit.

    \item Items $e_j$ with $j > \lceil \frac{1}{b_1} \rceil$. For these items, $A^3$ allocated them in a proportional way ($b_i$ fraction to agent $i$). Hence the fractions allocated in $A_4$ are 
    
    $$\sum_{i=1}^n s(B_i) \cdot b_i = \sum_{i=1}^n s(\sum_{j=1}^i b_j) \cdot b_i \ge \int_0^1 s(x) \; dx \ge 1$$
    
    
    To verify that the first inequality above, observe that $\sum_{i=1}^n s(\sum_{j=1}^i b_j) \cdot b_i$ is the area under a histogram with bars of width $b_i$ and height $s(\sum_{j=1}^i b_j)$, with total width $\sum_1^n b_i = 1$. As $s(x)$ is monotone non-decreasing, the height of the histogram at every point $x$ is at least $s(x)$.

    \item Items $e_k$ with $n < k \le \lceil \frac{1}{b_1} \rceil$. For each of these items, the allocation in $A^3$ is no longer proportional (in particular, agent~1 does not get a $b_1$ fraction of these items), but it dominates a proportional allocation in the sense that for every suffix of agents, the total fraction of the item allocated to these agents  under $A^3$ is at least as large as it is under the proportional allocation. (This is a consequence of the fact that in the process of creating $A^3$ from $A^2$, a surplus that covers a deficit always involves a surplus of a later agent covering a deficit of an earlier agent.) As later agents are scaled by larger fractions than earlier agents, it follows that the extent to which item $e_k$ is covered in $A^4$ is at least as large as items $e_j$ with $j > \lceil \frac{1}{b_1} \rceil$, 
    implying that $e_k$ has no deficit.
\end{enumerate}
\end{proof}

We next show that $A^5$ is an $(\bar{\alpha}, \bar{\beta})$ fractional allocation with $\max_i[\alpha_i + \beta_i] \le \rho$, for some $\rho < 2$. But first, we give some intuition regarding why such a statement may be true.

For the fractional allocation $A^1$, every agent $i$ receives a fraction $b_i$ of item $e_1$. Hence $\alpha_i = 1$.  As $1 < b_i < 1$, we have that $f_i = 1$. If $v_i(e_1) = 1$, then $\beta_i = 1$ as well, and then $\alpha_i + \beta_i = 2$, which is too large.

Allocation $A^2$ has the property that $f_i > \lfloor \frac{1}{b_i} \rfloor$ for every agent $i$. Consequently, by the properties of the chore share we have that $v_i(e_{f_i}) \le \frac{1}{2}$, and consequently $\beta_i \le \frac{1}{2}$. However, the total fractional value of allocation $A^2$ to agent $i$ might be larger than that of $A^1$, and consequently it might be that $\alpha_i > 1$. Hence, it is not clear whether $\alpha_i + \beta_i \le \rho < 2$ for all agents. The value of $\alpha_i$ with respect to $A^2$ is one aspect that we shall need to address in our proof. Here we make a qualitative observation, which is that the difficulties with $\alpha_i$ only arise for small values of $i$, but not for large values of $i$. The extremes are $\alpha_1$ which might be as large as $2 - b_1$ (if $v_1(e_1) = 1$ and all remaining items are equally valuable, with their number tends to infinity), which approaches~2 for small $b_1$, and $\alpha_n$ which cannot be larger than~1 (because for agent $n$, all items for which fractional values were decreased in $A^{2f}$ are at least as valuable than $e_n$).

Allocation $A^3$ causes no difficulties, because agents only replace fractions of items by fractions of less valuable items. Hence neither the value of $\alpha_i$ not the value of $\beta_i$ increases compared to their value in $A^2$.

If $s(x)$ is monotone non-decreasing with $\int_0^1 s(x) = 1$, then allocation $A^4$ modifies $A^3$ by decreasing the fractional value for the early agents (those are the agents that might enter this stage with large $\alpha_i$ and hence need it to be decreased) and increasing the fractional value for the late agents (those are the agents that enter this stage with small $\alpha_i$ and can afford to have it increased). This is done in such a way that no item suffers a deficit (as is proved in Lemma~\ref{lem:noDeficit}). In our proof we show that for every agent $i$, the combined effect of steps~2 and~4 is good for the agent. That is, all agents improve $\beta_i$ in step~2. The early agents suffer an increase of $\alpha_i$ in step~2, but they are compensated for it in step~4 to a sufficient extent, leading to $\alpha_i + \beta_i \le \rho$. The late agents are very happy after step~2 (the inequality $\alpha_i + \beta_i \le \rho$ is satisfied with slackness), and the loss that they suffer in step~4 is not larger than the slackness.

\begin{lemma}
\label{lem:alphabeta}
Suppose that $s(x)$ is monotone non-decreasing with $s(1) \le 2$, and $s(x)$ satisfies both $(1 - x)s(x) \le 1$ and $1 + s(x) - x s(x) \le s(1)$, for every $0 \le x \le 1$. Then in allocation $A^5$ the inequality $\alpha_i + \beta_i \le 1 + \frac{s(1)}{2}$ holds for every agent $i$, for every additive valuation function over indivisible chores.
\end{lemma}

\begin{proof}
Let $k = \max[n, \; \lfloor \frac{1}{b_i} \rfloor] + 1$. Then with respect to $A^5$, $f_i \ge k$. Namely the first item that $A^5$ allocates to agent $i$ in a strictly fractional manner (neither~0 nor~1) comes no earlier than $e_k$. Consequently $\beta_i \le v_i(e_k)$. Recall that $v_i(e_k) \le \frac{1}{2}$ (because $k \ge \; \lfloor \frac{1}{b_i} \rfloor + 1$). 

Consider first the case that $i \ge \lceil \frac{1}{b_i} \rceil$. (This necessarily holds for $i = n$, and maybe also for some smaller $i$.) In this case, the value of $A^{2f}$ to agent $i$ is no larger than $1 - v_i(e_i)$ (because every item whose fractional value was reduced in $A^{2f}$ comes no later than $e_i$). Consequently, the fractional value of $A^4$ is no larger than $v_i(e_i) + s(B_i)(1 - v_i(e_i)) \le v_i(e_i) + s(1)(1 - v_i(e_i))$. As $\beta_i \le v_i(e_k)$ we have that $\alpha_i + \beta_i \le s(1) + (1 - s(1))v_i(e_i) + v_i(e_k) \le s(1) + (2 - s(1))v_i(e_k)$ (the last inequality holds because $s(1) \ge 1$, and because $k > i$ implies that $v_i(e_k) \le v_i(e_i)$). As $v_i(e_k) \le \frac{1}{2}$ we have that $\alpha_i + \beta_i \le 1 + \frac{s(1)}{2}$. 

Consider now the case that $i < \lceil \frac{1}{b_i} \rceil$. (This necessarily holds for $i = 1$, and maybe also for some larger $i$.) In this case, the value of $A^{2f}$ to agent $i$ is no larger than $1 - i \cdot b_i \cdot v_i(e_i) - (\frac{1}{b_i} - i) \cdot b_i \cdot v_i(e_k) = 1 - v_i(e_k) - ib_i(v_i(e_i) - v_i(e_k))$. Consequently, the fractional value of $A^4$ is no larger than $v_i(e_i) + s(B_i)(1 - v_i(e_k) - ib_i(v_i(e_i) - v_i(e_k)))$. As $B_i \le ib_i$ and $v_i(e_i) \ge v_i(e_k)$ we have that $\alpha_i \le v_i(e_i) + s(B_i)(1 - v_i(e_k) - B_i(v_i(e_i) - v_i(e_k)))$. Using $\beta_i \le v_i(e_k)$ and rearranging we have:

$$\alpha_i + \beta_i \le s(B_i) + (1 - B_i s(B_i)) v_i(e_i) + (1 - s(B_i) + B_i s(B_i)) v_i(e_k)$$

As $(1 - x)s(x) \le 1$ for every $0 \le x \le 1$, the worst case is when $v_i(e_k)$ is as large as possible (and recall that we have the constraint $v_i(e_k) \le \frac{1}{2}$). If $1 \ge B_i s(B_i)$ then the worst case requires also $v_i(e_i)$ to be as large as possible (and recall that we have the constraint $v_i(e_i) \le 1$). We get that $\alpha_i + \beta_i \le s(B_i) + (1 - B_i s(B_i)) + \frac{1}{2} (1 - s(B_i) + B_i s(B_i)) = 1 + \frac{1}{2}(1 + s(B_i) - B_i s(B_i)) \le 1 + \frac{s(1)}{2}$ (using our assumption that $1 + s(x) - x s(x) \le s(1)$ for all $0 \le x \le 1$). If $1 < B_i s(B_i)$ then the worst case requires  $v_i(e_i)$ to be as small as possible, but we have the constraint $v_i(e_i) \ge v_i(e_k)$. Setting equality in that constraint, we get that $\alpha_i + \beta_i \le s(B_i) + (2 - s(B_i)) v_i(e_k)$. This is maximized (subject to the constraints $v_i(e_k) \le \frac{1}{2}$ and $B_i \le 1$, and the assumption that $s(1) \le 2$) when $v_i(e_k) = \frac{1}{2}$ and $B_i = 1$, giving $\alpha_i + \beta_i \le 1 + \frac{s(1)}{2}$.
\end{proof}

Given the conditions imposed on $s(x)$ in Lemmas~\ref{lem:noDeficit} and~\ref{lem:alphabeta}, we choose $s(x)$ to be of the following form, with $1 \le t \le 2$ to be chosen later:

\begin{equation*}
 s(x) =
    \begin{cases}
      \frac{t-1}{1-x} & 0\le x\le \frac{1}{t}\\
      t &  \text{for } x>\frac{1}{t}\\
    \end{cases}       
\end{equation*}

One may easily see that $s(x)$ is monotone non-decreasing in the range $0 \le x \le 1$, that $f(1) \le 2$ (for $t \le 2$). Lemma~\ref{lem:alphabeta} requires that $(1 - x)s(x) \le 1$, and this holds because $(1 - x)s(x) \le t - 1 \le 1$.  Lemma~\ref{lem:alphabeta} further requires that $1 + s(x) - x s(x) \le s(1)$. For $x \le \frac{1}{t}$ this translates to $1 + t - 1 \le t$ which holds with equality, and for $x > \frac{1}{t}$ this translates to $1 + t - xt \le t$ which holds with strict inequality.  

Lemma~\ref{lem:noDeficit} requires that $\int_0^1 s(x) \; dx \ge 1$, or equivalently, $(1-t)\ln(1-\frac{1}{t})+t-1 \ge 1$. Solving the inequality numerically we get that it holds when $t\ge 1.466$. As Lemma~\ref{lem:alphabeta} shows an approximation ratio is $1 + \frac{s(1)}{2} = 1+t/2$, we shall choose $t$ as small as possible, namely, $t = 1.466$.

We can now prove Theorem~\ref{thm:arbitrary}.


\begin{proof}
Choose $s(x)$ as above with $t = 1.466$, and recall that $s(1) = t$. As shown above, this choice of $s(x)$ satisfies all conditions of Lemmas~\ref{lem:noDeficit} and~\ref{lem:alphabeta}. Hence $A^5$ is a legal fractional allocation, and in $A^5$, the inequality $\alpha_i + \beta_i \le 1 + \frac{s(1)}{2} = 1.733$ holds for every agent $i$. By Lemma~\ref{lem:FracToSequence}, $A^5$ can be transformed (in polynomial time) into a picking sequence satisfying the guarantees of the theorem.
\end{proof}


{
\section{Picking sequences for equal entitlements}

In~\cite{aziz2022approximate} it was shown that there is a periodic picking sequence for indivisible chores that gives every agent a bundle of disvalue at most $\frac{5}{3}$ times her MMS (with better ratios for small $n$). Here we propose non-periodic picking sequences, prove that they offer a ratio no worse than $\frac{8}{5}$, and provide evidence that their ratio is no worse than $1.543$. Moreover, all our upper bounds on the approximation ratios hold compared to the chore share, and not just compared to the MMS. We complement our positive results by some lower bounds, and they hold compared to the MMS, and not just compared to the chore share.

\begin{theorem}
Consider allocation of indivisible chores to agents with additive valuation function and equal entitlements. For every $n$ and every $m$ there is a (polynomial time computable) picking sequence in which every agent that follows the greedy picking strategy gets a bundle of disvalue at most $\frac{8}{5}$ times her chore share (and thus also at most $\frac{8}{5}$ times her MMS). For sufficiently large $n$ and $m$, for every picking sequence, there is a way of assigning additive valuation functions to the agents so that if agents follow the greedy picking strategy, at least one agent gets a bundle of disvalue at least $\frac{3}{2}$ times her MMS. 
\end{theorem}

Agents are ordered from~1 to~$n$ arbitrarily.
We shall assume that the instance is an IDO instance, an assumption that can be made without loss of generality (see Section~\ref{sec:IDO}). The items are ordered from highest value $e_1$ to lowest value $e_m$. 
For each $n$, our respective picking sequence will be designed for an infinite number of items, and for every finite number of items $m$, one can take the prefix of length $m$ of our sequence.

\subsection{Notation, conventions, and picking orders}
\label{sec:order}

{We use $v_i$ (rather than $c_i$) to denote the disvaluation function of agent $i$.}
As noted above, we may assume that 
the instance is IDO ($v_i(e_j) \ge v_i(e_k)$ for every agent $i$ and every $j < k$). 
Recall that the chore share is $\max[v_i(e_1), v_i(e_n) + v_i(e_{n+1}), \frac{1}{n}\sum_{j= 1}^m v_i(e_j)]$. Scaling $v_i$ so that the chore share is~1, we have that $v_i(e_1) \le 1$, $v_i(e_{n+1}) \le \frac{1}{2}$, and $\sum_{j\le m} v_i(e_j) \le n$. We show that under such conditions our picking sequences never give an agent a bundle of value more than $\rho$, implying that $\rho$ is (an upper bound on) the approximation ratio of the allocation compared to the chore share. The approximation ratio compared to the MMS is at least as good, as the MMS is at least as large as the chore share.

We shall use the following notation.  For $n$ agents and $m$ items and a picking sequence $S$,  $\rho_{n,m}(S)$ denotes the approximation ratio of $S$ compared to the MMS. When comparing to the chore share, we shall use the notation $r_{n,m}$ instead. Observe that $\rho_{n,m}(S) \le r_{n,m}(S)$. We use $\rho_{n,m}$ to denote $\min_S[\rho_{n,m}(S)]$, the approximation ratio of the best picking sequence. We note that for every fixed $n$, both sequences $r_{n,m}$ and $\rho_{n,m}$ are nondecreasing as a function of $m$, as we can always pad a set of $m$ items by additional items of value~0. Hence $sup_m[\rho_{n,m}] = \lim_{n \rightarrow \infty} \rho_{n,m}$ (and similarly for $r_{n,m}$). We use $\rho_n$ to denote $sup_m[\rho_{n,m}]$ and $r_n$ to denote $sup_m[r_{n,m}]$.

{In describing our picking sequences, we shall use the correspondence between picking sequences and allocations for IDO instances, as described in Corollary~\ref{cor:IDO}. For convenience, we shall describe these allocations as if they are the result of what we shall refer to as {\em picking orders}. In contrast to picking sequences, in picking orders we assume that agents pick in their turns the worst possible chores (instead of best possible chores). The picking sequences that correspond to these allocations under Corollary~\ref{cor:IDO} are not the picking orders, but rather the reverse of these picking orders (e.g., the agent who picks in round~1 in the picking order picks in round~$m$ in the picking sequence).} 

\subsection{Ridge picking orders}
\label{sec:ridge}


{Our allocations can be described by} picking orders that start with what we shall refer to as a {\em ridge}. Namely, the first $n$ items are given to agents~1 to~$n$ (in increasing order), and the next $n$ items are given to agents $n$ to~1 (in decreasing order). In other words, every agent $i$ gets items $e_i$ and $e_{2n-i+1}$. We refer to picking orders that start with a ridge as {\em ridge picking orders}. When restricting attention to ridge picking orders, we use notation of $\hat{\rho}$ and $\hat{r}$ instead of $\rho$ and $r$. Note that $\hat{\rho}_n \ge \rho_n$ and $\hat{r}_n \ge r_n$.

We shall show how, given a value of $n$, one can design the ridge picking order $S$ for which $\hat{r}_{n}(S) = \hat{r}_{n}$. Namely, this $S$ is the ridge picking order with best approximation ratio compared to the chore share. For fixed small $n$, this will allow us to determine $\hat{r}_n$ exactly, whereas for large $n$, we shall obtain upper bounds on $\hat{r}_n$. Of course, these upper bounds hold also for $r_n$ and for $\rho_n$.

Our approach also allows to determine $\hat{r}_{n,m}$ (for finite $m$), but our presentation will focus on the more difficult case of $\hat{r}_n$ in which $m$ is not bounded, leading to ridge picking orders of infinite length. 


In our ridge picking orders, every agent $i$ is associated with a period $p_i$ for picking items. This period need not be an integer. The period $p_i$ determines a sequence of thresholds on the rounds in which an agent is allowed to pick items. That is, agent $i$ is not allowed to pick her $t$th item at a round earlier than the $t$th threshold in her sequence of thresholds. The sequence of thresholds for the agents are listed below. For this purpose, the agents are partitioned into three classes. 

\begin{itemize}
    \item Class~0 contains agents of intermediate values of $i$, and for them the period starts immediately. 

$$\lceil p_i \rceil, \; \lceil 2p_i \rceil, \; \lceil 3p_i \rceil, \ldots$$

Importantly, two {\em ridge constraints} will hold for every agent $i$ in class~0. The constraints are $i \ge \lceil p_i \rceil$ and $2n-i+1 \ge \lceil 2p_i \rceil$ (as $i$ is an integer, we may remove the ceiling notation), indicating that the ridge does not violate the sequence of thresholds.

\item Class~1 contains agents of small values of $i$, and for them the period starts after they get their first item.

$$i, \; i+\lceil p_i \rceil, \; i+\lceil 2p_i \rceil, \; i+\lceil 3p_i \rceil, \ldots$$

For every agent $i$ in class~1 one ridge constraint is required to hold. The constraint is $2n-i+1 \ge i + \lceil p_i \rceil$, indicating that the ridge does not violate the sequence of thresholds.
    
    \item Class~2 contains agents of high values of $i$, and for them the period starts after they get their second item.

$$i, \; 2n-i+1, \; 2n-i+1+\lceil p_i  \rceil, \; 2n-i+1+\lceil 2p_i  \rceil, \; 2n-i+1+\lceil 3p_i \rceil, \ldots$$

\end{itemize}

Our choice of periods $p_i$ will need to ensure that that all items can be  allocated, without violating any of the thresholds. Equivalently, it will need to ensure that for every $j$, the number of thresholds of value at most $j$ in all lists is at least $j$. We refer to this as the {\em covering constraints}. Observe that for $j \le 2n$ the respective covering constraint holds (due to the ridge constraints). We will need to ensure that the covering constraints holds also for $j > 2n$. A necessary requirement for the covering constraint to hold (when $m$ tends to infinity) is that $\sum_{i=1}^n \frac{1}{p_i} \ge 1$. We refer to this as the {\em fractional covering constraint}.


Recall that $v_i$ denotes the disvaluation function of agent $i$, and that her chore share is assumed to be~1. 
In addition to the covering constraints, we have for each agent $i$ a constraint $F_i$, specifying that the agent got a bundle of cost at most $\rho$.  Define $x_i = v_i(e_i)$, $y_i = v_i(e_{2n-i+1})$. Note that necessarily $x_i \le 1$ and $y_i \le \frac{1}{2}$, and that $x_i \ge y_i$. We also define $a_i = \sum_{j=1}^i v_i(e_j)$ (the disvalue of the prefix), $b_i = \sum_{j=i+1}^{2n-i+1} v_i(e_j)$ (the disvalue of the middle portion) and $c_i = \sum_{j=2n-i+2}^m v_i(e_j)$ (the disvalue of the suffix). Then 
$a_i + b_i + c_i \le n$.

The expression for (an upper bound on) the cost of the bundle received by agent $i$ depends on the class of the agent. This affects the form of the constraints $F_i$. For {agent $i$ in} class~0, the constraint {$F_i$} is $\frac{n}{p_i} \le \rho$, for {agent $i$ in} class~1 {the constraint} {$F_i$} is $x_i + \frac{b_i + c_i}{p_i} \le \rho$, and for {agent $i$ in} class~2 {the constraint} {$F_i$}
is $x_i + y_i + \frac{c_i}{p_i} \le \rho$. 
In particular, as $c_i$ might be arbitrarily close to~$n$ (as $a_i$ and $b_i$ can be arbitrarily small, if $m$ is sufficiently large), the constraint $p_i \ge \frac{n}{\rho}$ is implied. In general, the constraint $F_i$ pushes the value of $p_i$ to be large, whereas the covering constraints push the $p_i$ values to be small. We wish to find the smallest value of $\rho$ for which these two contradicting goals are feasible.

To test whether an approximation ratio of $\rho$ is feasible (for a particular number $n$ of agents), we first extract from each constraint $F_i$ a value of $p_i$ that ensures that the constraint $F_i$ holds for the particular value of $\rho$. Thereafter, we check if the covering constraints are satisfied with these $p_i$ values. If they are, then $\rho$ is feasible. Else, $\rho$ is not feasible (for ridge picking orders). In this paper, we shall present the analysis for $m$ tending to infinity. We remark that the same principles can be used for any fixed finite $m$, and the value of $\rho$ would be no larger (and possibly smaller) than the value derived for infinite $m$.

Let us now explain how we determine which agents belong to which class, and how $p_i$ is computed for each class.

\begin{itemize}
    \item Class~0. Here we set $p_i = \frac{n}{\rho}$. Recall that the ridge constraints require that $i \ge  p_i $ and $2n-i+1 \ge 2p_i$. Hence an agent $i$ belongs to class~0 if $ \frac{n}{\rho} \le i \le 2n + 1 -  \frac{2n}{\rho}$.
    
    \item Class~1. Here $i \le  \frac{n}{\rho}$. 
    {To give some intuition for our choice of $p_i$ in this case, suppose that $x_i = 1$ and that $y_i$ is negligible (we shall soon explain why these are worst case values for $x_i$ and $y_i$, for our choice of $p_i$).} 
    Then $F_i$ becomes $1 + \frac{n - i}{p_i} \le \rho$, implying that $p_i \ge \frac{n - i}{\rho - 1}$. We choose the smallest possible $p_i$ (with an eye towards satisfying the covering constraint), namely, $p_i = \frac{n - i}{\rho - 1}$. Note that for this choice of $p_i$ (and $i \le \frac{n}{\rho}$) we have that $p_i \ge i$. This justifies making $x_i$ as large as possible, as from the prefix we select with higher frequency than from the suffix. The ridge constraint for class~1 
    is $p_i  \le (2n - i + 1) - i$. For $p_i = \frac{n - i}{\rho - 1}$ the ridge constraint becomes $\rho \ge \frac{3n-3i+1}{2n-2i+1}$. For large $n$ we shall have $\rho \ge \frac{3}{2}$ and then the ridge constraint holds. (For small $n$ the ridge constraint will be checked by hand.) The ridge constraint justifies making $y_i$ as small as possible, as from the middle portion we select with lower frequency than from the suffix.
    
    \item Class~2. Here $i > 2n + 1 -  \frac{2n}{\rho}$. 
    {To give some intuition for our choice of $p_i$ in this case, suppose that $x_i = y_i$ and that $y_i$ is as large as possible (we shall soon explain why these are worst case values for $x_i$ and $y_i$, for our choice of $p_i$).} 
    The largest possible value for $y_i$ is $\frac{1}{2}$. Then $F_i$ becomes $1 + \frac{i-1}{2p_i} \le \rho$, implying that $p_i \ge \frac{i-1}{2(\rho - 1)}$, {and we choose $p_i \ge \frac{i-1}{2(\rho - 1)}$}. From the combination of the prefix and middle we select at a higher frequency than from the suffix, justifying making $y_i$ as large as possible. {From the prefix we select at a lower frequency than from the suffix, justifying making $x_i$ as large as possible (conditioned on first maximizing $y_i$).}
    
    
\end{itemize}

}

{
\subsection{A lower bound on $\rho$}\label{subsec-largen}

In this section we prove Theorem~\ref{thm:lowerbound}.

A family of picking orders (one for each value of $n$) referred to as Sesqui Round Robin (SesquiRR) is considered in~\cite{aziz2022approximate}. It achieves $\rho_2(SesquiRR) = \frac{4}{3}$ and $\rho_3(SesquiRR) = \frac{7}{5}$. Both ratios are best possible. For $n=4$ it achieves $\rho_4(SesquiRR) = \frac{3}{2}$. It is claimed in~\cite{aziz2022approximate} {that SesquiRR is not optimal for $n \ge 4$, and} that $1.405 < \rho_4 < 1.499$ (the proof is omitted from that paper). For larger value of $n$, it is proved that  $\rho_n(SesquiRR) \le \frac{5}{3}$. We note that for $n \le 3$ SesquiRR is a ridge picking order, but for $n \ge 4$ it is not. 

\begin{proposition}
\label{pro:ridge}
For every $n \ge 2$ and $m \ge 2n$, if a picking order $S$ is not a ridge picking order, then $\rho_{n,m}(S) \ge \frac{3}{2}$.
\end{proposition}

\begin{proof}
For a given value of $n$ and $m \ge 2n$, consider a picking order $S$ for which $\rho_{n,m}(S) < \frac{3}{2}$.

It cannot be that in $S$ some agent picks two items among the first $n$, as then her approximation ratio will be~2 if each of the first $n$ items has disvalue~1. Hence we may assume that for every $i \le n$, agent $i$ picks item $i$.

It cannot be that in $S$ some agent picks three items among the first $2n$, as then her approximation ratio will be~$\frac{3}{2}$ if each of the first $2n$ items has disvalue~$\frac{1}{2}$. Hence we may assume that each agent picks only one item from items $\{e_{n+1}, \ldots, e_{2n}\}$.

It cannot be that in $S$, for some $i$, the second item picked by agent $i$ is strictly earlier than $2n - i + 1$, as then her approximation ratio will be~$\frac{3}{2}$ if each of the first $i$ items has value~1, and the next $2(n-i)$ items each has value~$\frac{1}{2}$.

The above implies that $S$ is a ridge picking order. This can be proved by observing that the only agent that can pick item $e_{n+1}$ is agent $n$, then the only agent that can pick $e_{n+2}$ is agent $n-1$, and so on.
\end{proof}

\begin{corollary}
\label{cor:ridge}
If $\rho_{n,m} \le \frac{3}{2}$, then $\rho_{n,m} = \hat{\rho}_{n,m}$. In other words, ridge picking orders are optimal for those values of $n$ (and $m$) for which at approximation ratio no worse than $\frac{3}{2}$ is possible.
\end{corollary}

\begin{proposition}
\label{pro:largen}
For every $\rho \le 1.524$, there is a sufficiently large $n_{\rho}$ such that for every $n \ge n_{\rho}$, $\hat{\rho}_n \ge \rho$. That is, for sufficiently large $n$, ridge picking orders do not offer an approximation ratio better than $1.524$ (when $m$ is sufficiently large).
\end{proposition}

\begin{proof}
We show that for large $n$, if $\rho$ is too small then the implied periods $p_i$ do not satisfy the fractional covering constraint. This implies that as $m$ tends to infinity, the number of items picked by at least one agent $i$ is at least $\frac{m}{p'_i}$, where $p'_i \le  p_i - \epsilon$, for some $\epsilon$ that depends only on $\rho$. We then fit a valuation function $v_i$ as follows.  In $v_i$, if $i$ is of class~0 than all items have the same value, if $i$ is of class~1 then the first $i$ items each has value~1 and the remaining items all have the same value, and if $i$ is of class~2 then the first $2n-i+1$ items each has value~$\frac{1}{2}$ and the remaining items all have the same value. For this $v_i$ the MMS is at most $1 + O(\frac{1}{m})$, but agent $i$ gets a bundle of disvalue at least $\rho(1 + \Omega(\epsilon))$, which is larger than $\rho(1 + O(\frac{1}{m}))$ when $m$ is sufficiently large.

The assumption that $n$ is very large allows us to approximate the sum in the fractional covering constraint by an integral. We shall denote $\frac{i}{n}$ by $\alpha$, and so $i = \alpha n$. We shall also assume that $\rho \ge \frac{3}{2}$, and so we know that the ridge constraint for class~1 holds. (The analysis indeed gives $\rho > \frac{3}{2}$, confirming this assumption.)

Approximating the fractional covering constraint by an integral, doing integration by parts, and omitting terms that effect the end result only by an $O(\frac{1}{n})$ additive term (this includes omitting $O(1)$ additive terms from the periods, and terms of order $O(\frac{1}{n})$ from ranges of integration), we can rewrite the fractional covering constraint as:

$$\int_0^{\frac{1}{\rho}} \frac{\rho - 1}{1 - \alpha} d\alpha + \int_{\frac{1}{\rho}}^{\frac{2(\rho - 1)}{\rho}} \rho d\alpha
+ \int_{\frac{2(\rho - 1)}{\rho}}^1 \frac{2(\rho-1)}{\alpha} d\alpha \ge 1$$

The first summand is the contribution of class~1, the second is the contribution of class~0, and the third is the contribution of class~2.

After integration this gives:

$$(\rho-1)\ln \frac{\rho}{\rho - 1} + 2\rho - 3 + 2(\rho - 1)\ln \frac{\rho}{2(\rho - 1)} \ge 1$$


By WolframAlpha, the smallest value of $\rho$ that satisfies the above is $\rho \simeq 1.52408$.

\end{proof}


The following corollary is a restatement of Theorem~\ref{thm:lowerbound}.

\begin{corollary}
\label{cor:largen}
There is some $n_0$ such that for every $n \ge n_0$, $\rho_n \ge \frac{3}{2}$. In other words, if $n$ and $m$ are sufficiently large, then there is no picking sequence with an approximation ratio better than $\frac{3}{2}$.
\end{corollary}

\begin{proof}
Follows by combining Proposition~\ref{pro:ridge} and Proposition~\ref{pro:largen}.
\end{proof}
}

{\subsection{Very small $n$}
\label{sec:smalln}


For $n = 2$ our ridge picking order gives periods $p_1 = 3$ and $p_2 = \frac{3}{2}$. The ridge picking order is $1221(221)^*$, and the approximation ratio is $\frac{4}{3}$.

For $n = 3$ our ridge picking order gives periods $p_1 = 5$ and $p_2 = p_3 = \frac{5}{2}$. The ridge picking order is $123321(23321)^*$, and the approximation ratio is $\frac{7}{5}$.

In both $n=2$ and $n=3$ the picking orders and approximation ratios are exactly as in~\cite{aziz2022approximate}, and are optimal.

For $n=4$ the SesquiRR picking order presented in~\cite{aziz2022approximate} is $(123443)^*$, giving a ratio of $\frac{3}{2}$. It is claimed in~\cite{aziz2022approximate} that there are other picking sequences giving a ratio better than $1.499$, and that no picking sequence can give a ratio better than $1.405$. We show that ridge picking orders get a ratio better than $1.499$, even compared to the chore share.

\begin{proposition}
\label{pro:r4}
For picking sequences over indivisible chores, $\hat{r}_4 = r_4 = \frac{13}{9} < 1.445$, whereas $\rho_n \ge \frac{10}{7} > 1.428$. In other words, a ridge picking order (explicitly given) achieves an approximation ratio of $\frac{13}{9}$ compared to the chore share, and no picking sequence can achieve an approximation ratio better than $\frac{10}{7}$ compared to the MMS.
\end{proposition}

\begin{proof}
As $r_4 < \frac{3}{2}$, ridge picking orders attain the best approximation ratios among picking sequences. 

Let $r$ be the approximation ratio of the ridge picking order. Then the constraints that we derive are based on agents~1 and~2 being of class~1, agent~3 being of class~0, and agent~4 being of class~2.

\begin{itemize}

    \item $1 + \frac{3}{p_1} = r$.

\item $1 + \frac{2}{p_2} = r$.

\item $\frac{4}{p_3} = r$. (This will turn out to imply also that $1 + \frac{1}{p_3} = r$, and hence agent~3 is not of class~3.) 

\item $1 + \frac{3}{2p_4} = r$.

\end{itemize}

Solving together with the fractional covering constraint ($\sum \frac{1}{p_i} = 1$) we get $r = \frac{10}{7}$, $p_1 = 7$, $p_2 = \frac{14}{3}$, $p_3 = \frac{14}{5}$ and $p_4 = \frac{7}{2}$. For each agent this gives a sequence of thresholds, but unfortunately, it cannot be translated to a schedule that covers all items. Specifically, the covering constraint for $j=11$ is not satisfied.  After $12344321$, in the next 3 locations we cannot put agents~1 or~2, and we can put each of agents~3 and~4 only once. Hence the item at location~11 cannot be covered without decreasing one of the periods. Choosing to decrease the period for agent~2 (this turns out to be optimal), we get the schedule $12344321\; (4324331\; 4324321)^*$, with an approximation ratio no worse than $1 + \frac{4}{9}$. The worst case for this schedule with respect to the chore share is when for agent~2 the first two items are each worth~1, and the next nine items are each worth $\frac{2}{9}$. 
\end{proof}

}


{
\subsection{Upper bounds on $\hat{r}_n$}
\label{sec:85}

Here we develop upper bounds on $\hat{r}_n$ that hold simultaneously for all $n$. Much of the difficulty in developing such bounds stems from the covering constraints. The proof of Proposition~\ref{pro:largen} indicates that dealing with the fractional covering constraint is tractable. However, as is evident from the proof of Proposition~\ref{pro:r4}, the (integral) covering constraints are strictly more demanding than the fractional covering constraint. For every fixed $n$, we can use an approach similar to that of Section~\ref{sec:smalln} in order to satisfy the covering constraints and determine $\hat{r}_n$. Here we develop a technique that handles the covering constraints for all values of $n$ simultaneously. It is based on a combination of Lemma~\ref{lem:doublen} and Lemma~\ref{lem:8r}.

\begin{lemma}
\label{lem:doublen}
For ridge picking orders and every $n \ge 4$, it holds that $\hat{r}_{n} \le \hat{r}_{2n}$.
\end{lemma}

\begin{proof}
Let $r= \hat{r}_{2n}$. Then for every agent $i$ ($1 \le i \le 2n$), this $r$ determines a period $p_i$ as in Section~\ref{sec:ridge}, such that with these $p_i$ values the (integral) covering constraint holds. Each $p_i$ determines a sequence {$t_i$} of thresholds $\{t_i^1, t_i^2, t_i^3, \dots\}$ such that agent $i$ is not allowed to pick her $j$th item before round $t_i^j$. Importantly, these sequences of thresholds have the {\em domination property} that for every $i$, either for all $j \ge 3$ it holds that $t_{2i - 1}^j \ge t_{2i}^j$ (e.g., this happens if $2i-1$ is of class~1), or for all $j \ge 3$ it holds that $t_{2i - 1}^j \le t_{2i}^j$ (e.g., this happens if $2i$ is of class~2). Note that for $n \ge 4$, the number $2n$ of agents is such that class~0 is non-empty, and hence there is no $i$ for which $2i-1$ is of class~1 and $2i$ is of class~2.

Using these sequences of thresholds (that apply to a ridge picking order over $2n$ items), we set values for the thresholds $\tau_i^j$ for an infinite ridge picking order over $n$ items. Specifically, we set $\tau_i^j = \lceil \frac{1}{2}\min[t_{2i - 1}^j, t_{2i}^j] \rceil$. 
Observe that as the {$2n$} threshold sequences $(t_1, \ldots, t_{2n})$ come from a ridge picking order, the {$n$} threshold sequences $(\tau_1, \ldots, \tau_n)$ also allow for a prefix that is a ridge (namely, $\tau_i^1 \ge i$ and $\tau_i^2 \ge 2n - i + 1$). The domination property above, together with the fact that $\lceil \frac{1}{2}t_{2i - 1}^1 \rceil = \lceil \frac{1}{2} t_{2i}^1 \rceil$ and  $\lceil \frac{1}{2}t_{2i - 1}^2 \rceil = \lceil \frac{1}{2} t_{2i}^2 \rceil$,  implies that for every $i$, there is $i' \in \{2i - 1, 2i\}$ such that for every $j$, $\tau_i^j = \lceil \frac{1}{2}t_{i'}^{j}\rceil$.

We first show that $(\tau_1, \ldots, \tau_n)$ satisfy the covering constraint.  {For a round $j$ and an agent $i$, let $\#\tau_i^{\le j}$ denote the number of thresholds in sequence $\tau_i$ whose value is at most $j$.} For the sake of contradiction, assume that that $(\tau_1, \ldots, \tau_n)$ do not satisfy the covering constraint. Then there is an earliest round $j$ such that {$\sum_i \#\tau_i^{\le j} \le j - 1$}. Observe that $j > 2n$, because the threshold sequences $(\tau_1, \ldots, \tau_n)$ allow for a prefix that is a ridge. Consider now round $2j$ with respect to $(t_1, \ldots, t_{2n})$, {and let $\#t_i^{\le 2j}$ denote the number of thresholds in sequence $t_i$ whose value is at most $2j$. The fact that $2j$ is even and larger than $2n$ implies that for every $i\le n$, $2\#\tau_i^{\le j} \ge  \#t_{2i-1}^{\le 2j} + \#t_{2i}^{\le 2j}$. Hence $\sum_{i=1}^{2n} \#t_i^{\le 2j} \le 2\sum_{i=1}^n \#\tau_i^{\le j} \le 2j-2$, contradicting the assumption that the covering constraints hold for $(t_1, \ldots, t_{2n})$.} 

We now show that for every agent $i \le n$, the respective sequence $\tau_i$ of thresholds satisfies the constraint $F_i$ {(defined in Section~\ref{sec:ridge})}, when we substitute $r$ for $\rho$. 
Assume for the sake of contradiction that for some $i$, constraint $F_i$ does not hold. Namely, there is a disvaluation function $v_i$ (with $v_i(e_1) \le 1$, $v_i(e_{n+1}) \le \frac{1}{2}$ and $\sum_j v_i(j) \le n$) such that $\sum_{j} v_i(e_{\tau_i^j}) > r$. Recall that there is some $i' \in \{2i - 1, 2i\}$ such that for every $j$, $\tau_i^j = \lceil \frac{1}{2}t_{i'}^{j}\rceil$. Define a valuation function $v_{i'}$ satisfying $v_{i'}(e_{2j - 1}) = v_{i'}(e_{2j}) = v_i(e_j)$ for every $j$. This valuation function obeys the constraints for the chore share over $2n$ agents, namely $v_{i'}(e_1) \le 1$, $v_{i'}(e_{2n+1}) \le \frac{1}{2}$ and $\sum_j v_{i'}(j) \le 2n$. Moreover, $\sum_{j} v_{i'}(e_{t_i^j}) = \sum_{j} v_i(e_{\tau_i^j}) > r$, contradicting the assumption that $r= \hat{r}_{2n}$.
\end{proof}

\begin{corollary}
\label{cor:doublen}
For ridge picking orders, every $n \ge 4$ and every integer $k \ge 0$, it holds that $\hat{r}_{n} \le \hat{r}_{2^kn}$.
\end{corollary}

\begin{proof}
Apply Lemma~\ref{lem:doublen} consecutively $k$ times.
\end{proof}

\begin{lemma}
\label{lem:8r}
Let $n$ be divisible by~8. Then $\hat{r}_n \le \frac{8}{5}$.
\end{lemma}

\begin{proof}
Suppose that $n$ is divisible by~8, that is $n = 8\ell$. Partition the agents into~8 blocks of size $\ell$, named $\{a,b,c,d,e,f,g,h\}$. Each block can be thought of as a {\em super agent}. Likewise, after sorting chores from highest value to smallest value, the sequence of chores is also partitioned into blocks of size $\ell$, where each block is referred to as a {\em super item}. We shall design a picking order in which super agents pick super items. {We want this picking order for super agents to have the property that it can be lifted to a picking order with $\rho \le 1.6$ for the original instance. The lifting is done by giving a distinct item from the super item to every distinct agent from a the super agent that picked the item. For simplicity of the analysis (that is made possible by the fact that we take $\rho = 1.6$, whereas a smaller value of $\rho$ would probably suffice for a more complicated analysis), we assume a worst case scenario in which for every super agent, there is one agent within the super agent that always receives the item of highest disvalue (when allocating chores of a super item to the agents of the super agent). For this reason, we assume that $\ell$ tends to infinity, and take the period of the picking orders for the super agent to be that of the first agent within the super agent (even though later agents could have shorter periods).} For example, for super agent $a$ that takes super item~1, had we been in the situation that $n=8$ we would make the worst case assumption that item~1 has disvalue~1 and the disvalue left is~7. Then the period $p_a$ would satisfy $1 + \frac{7}{p_a} = \frac{8}{5}$, giving $p_a = \frac{35}{3}$. However, for $n = 8\ell$ for large $\ell$, the value taken by the first agent in the super agent is negligible. To accommodate for this, we take a longer period satisfying $1 + \frac{8}{p_a} = \frac{8}{5}$, namely, $p_a = \frac{40}{3}$.

The periods for our picking orders for every super agent are listed below. 

Super agents of class~1 (first pick super item $i$ and then start a period):

\begin{itemize}

\item $a$: $1$, $p_a = \frac{40}{3} < 14$. 

\item $b$: $2$, $p_b = \frac{35}{3} < 12$. 

\item $c$: $3$, $p_c = 10$. 

\item $d$: $4$, $p_d = \frac{25}{3} < 9$. 

\item $e$: $5$, $p_e = \frac{20}{3} < 7$. 

\item $f$: $6$, $p_f = 5$. (Though super agent $f$ appears to be of class~1, the agents making up super agent $f$ are in fact of class~0.)

\end{itemize}

Super agents of class~2 (first pick super items $i$ and $16 - i + 1$ and then start a period):

\begin{itemize}

\item $g$: $7,10$, $p_g = \frac{35}{6} < 6$. 

\item $h$: $8,9$, $p_h = \frac{20}{3} < 7$. 

\end{itemize}

{We claim that with the above periods, the covering constraints hold. As a sanity check, we first verify that the fractional covering constraint holds. 
As the sum of rates is $\frac{1473}{1400} > 1.05$, this is indeed the case. We now design a picking order based on the given periods.}

Our picking order has a periodic subsequence of length 40, starting at round~11 (namely, after all periods begin, including those of class~2). We list below the thresholds (up to round~50) for each (super) agent, given their own respective periods. 
In our periodic subsequence we enforce the restriction that the number of picks of an agent is not larger than her rate times~40, rounded {down} to the nearest integer. 
This restriction is needed so that we can indeed repeat the periodic subsequence. The actual rounds in which picks are made appear in parenthesis.

\begin{itemize}

\item $a$: $1$, $p_a = \frac{40}{3}$. $15 (16), 28 (29), 41 (43).$ 

\item $b$: $2$, $p_b = \frac{35}{3}$. $14 (15), 26 (27), 37 (39), 49.$ 

\item $c$: $3$, $p_c = 10$. $13 (14), 23 (25), 33 (35), 43 (45).$ 

\item $d$: $4$, $p_d = \frac{25}{3}$. $13 (13), 21 (21), 29 (31), 38 (40), 46.$ 

\item $e$: $5$, $p_e = \frac{20}{3}$. $12 (12), 19 (20), 25 (26), 32 (34), 39 (41), 45 (47).$  

\item $f$: $6$, $p_f = 5$. $11 (11), 16 (19), 21 (22), 26 (28), 31 (33), 36 (37), 41 (44), 46 (48)$. 

\item $g$: $7,10$, $p_g = \frac{35}{6}$. $16 (18), 22 (23), 28 (30), 34 (36), 40 (42), 45 (49).$ 

\item $h$: $8,9$, $p_h = \frac{20}{3}$. $16 (17), 23 (24), 29 (32), 36 (38), 43 (46), 49 (50).$ 

\end{itemize}

The agents picking in rounds~1 up to~50 are listed below, where the suffix of length~40 repeats itself indefinitely.

{\bf abcdefghhg $($fedcbahgfe dfghcebfag dhfecgfhbd egafchefgh$)^*$}

\end{proof}

\begin{theorem}
\label{thm:85}
For every $n$, $\hat{r}_n \le \frac{8}{5}$.
\end{theorem}

\begin{proof}
The case of $n \le 4$ is addressed in Section~\ref{sec:smalln}. For $n \ge 4$, Corollary~\ref{cor:doublen} implies that $\hat{r}_n \le \hat{r}_{8n}$, and Lemma~\ref{lem:8r} implies that $\hat{r}_{8n} \le \frac{8}{5}$.
\end{proof}

}

\subsection{Computer assisted analysis of approximation ratio}


Our proof of Theorem~\ref{thm:85} is based on designing picking orders for~8 super agents (Lemma~\ref{lem:8r}). The value of~8 was chosen as it is sufficiently large so as to give good bounds (a ratio of $\frac{8}{5}$) that improve over previously known results, yet not too large, and thus the analysis could be completed and verified by hand. However, it is clear that as the number of super agents grows, the bounds will improve further. In this section we use computer assisted analysis to explore to what extent they can be improved. We find that the best approximation ratio that our picking orders can give compared to the chore share is somewhere between 1.542 and 1.543.  We describe here how we performed the computer assisted analysis, and report the results that we got. We encourage interested readers to perform independent confirmation of our results, by using the principles described here to write their own code and run it. 

{The proof of Lemma~\ref{lem:8r} designs a covering sequence that is ultimately periodic. In contrast, the computer program tests whether for a given target approximation ratio $\rho$, there is a sequence that covers the first $t$ super chores, where $t$ is chosen based on the number of super agents and $\rho$. We prove that if the picking order can cover first $t$ super chores for our choice of $t$, then it can be extended to cover any number of super chores (while maintaining an approximation ratio no worse than $\rho$). 

Recall that the target value $\rho$ dictates a period $p_i$ for each agent $i$. Among other constraints, these periods need to satisfy the fractional covering constraint, $\sum_{i} \frac{1}{p_i}\ge 1$. We define the covering ratio as $r=\sum_{i} \frac{1}{p_i}$. The larger $r$ is, the more slackness we have in the fractional covering constraint.

Given a target ratio $\rho$, let $P(i,k)$ be the total number of chores that agent $i$ can pick among the first $k$ chores (without violating $\rho$). A picking order exists if the (integer) {\em covering constraint} $\sum_{i}P(i,k)\ge k$ holds for every $k$. 

\begin{proposition}\label{prop:cover}
Suppose that there are $n$ super agents and target ratio $\rho$ for which the associated periods satisfy the fractional covering constraint with covering ratio $r>1$. If the (integer) covering constraints hold for all $k \le 2n+\frac{n}{r-1}$ chores, then there is a picking order with approximation ratio at most $\rho$. 
\end{proposition}

\begin{proof}
The first $2n$ chores are covered by default. Each agents start counting her period at some chore among the first $2n$ chores. Thus, for each agent $i$, we have $$P(i,k)\ge 2+\left\lfloor\frac{k-2n}{p_i}\right\rfloor\ge 1+\frac{k-2n}{p_i}.$$

When $k\ge 2n+\frac{n}{r-1}$, we have 
\begin{equation*}
    \begin{split}
        \sum_{i}P(i,k)&\ge n+\sum_{i}\frac{k-2n}{p_i}\\
        &= n+(k-2n)\cdot r\\
        &= k-n+(k-2n)\cdot (r-1)\\
        &\ge k-n+\frac{n}{r-1}\cdot (r-1)\\
        &=k
    \end{split}
\end{equation*}

This means that, when $k\ge 2n+\frac{n}{r-1}$, the associated covering constraint is satisfied.
Hence it suffices to check the first $2n+\frac{n}{r-1}$ covering constraints.
\end{proof}

Proposition \ref{prop:cover} implies that $t = 2n+\frac{n}{r-1}$ suffices. Consequently, we design the following algorithm. 
}

\begin{algorithm}
\KwIn{The number of super agents $n$, a target approximation ratio $\rho$}
\KwOut{Pass or Fail}
\BlankLine
    Based on $\rho$, compute the period $p_i$ of each super agent $i$\;
    Compute the covering ratio $r=\sum\frac{1}{p_i}$\;
    \For{$k\rightarrow 2n$ to $ 2n+\frac{n}{r-1}$}
    {
     Test the covering constraint $\sum_{i}P(i,k)\ge k$\;
     \If{Test Failed}{\Return {Fail}}
     }
    
	\Return {Pass}
\caption{Ratio test}\label{alg-test}	
\end{algorithm}

We ran the algorithm with $n=2^{14}=16384$ super agents and $\rho=1.543$. This gave values of $r\simeq 1.03448$ and $t\simeq 30.9991\cdot n \le 507890$. The result was {\em Pass}.  To verify that increasing the number of super agents will not significantly improve the value of $\rho$, we ran the algorithm with $n=2^{14}=16384$ agents (not super agents) and $\rho=1.542$. 
This gave $r\simeq 1.03277$ and $t\simeq 32.5166\cdot n\le 532752$. The result was {\em Fail} (at $k=42465$). 

\section{Envy in picking sequences}
\label{sec:envy}

In this section we consider picking sequences for agents with additive valuations, and items might be either goods or chores.

A picking sequence $\pi$ specifies for every label $i$ the set of {\em picking rounds} $R_i$ in which the agent who gets the label $i$ picks items. 
A picking sequence is an allocation mechanism. As such, it induces a game between the agents. A strategy $s_i$ for agent $i$ in such a game specifies which item to pick in each of her picking rounds, given the history of picks in all rounds up to that round. A picking sequence $\pi$ together with the vector $(s_1, \ldots, s_n)$ of strategies for the agents determines the resulting allocation {$B_1, \ldots, B_n$}. The value of this allocation to agent $i$ is $v_i(B_i)$. For an agent $i$, the choice of which strategy to use may depend on her valuation function $v_i$, on her beliefs concerning which strategies other agents will be using, and possibly on other factors (e.g., computational constraints, cognitive biases, random noise, etc.). In particular, as other agents are free to choose their own strategies, an agent $i$ who knows only $s_i$ and $\pi$ does not know which bundle she will get. Likewise, an agent $i$ that knows only $\pi$ and $v_i$ does not know which strategy $s_i$ will give her the best bundle.

We say that agent $i$ is {\em risk averse} if her goal is to maximize her {\em guaranteed} utility (equivalently, minimize her worst possible disutility, in the case of chores). That is, for every strategy $s_i$ the agent assumes that the vector of strategies that the other agents pick is such that it minimizes the utility of $i$, given that $i$ uses strategy $s_i$. Under this assumption, the agent wishes to pick a strategy $s_i$ that maximizes her utility. Luckily, for picking sequences, the optimal risk averse strategy is straightforward (when agents have additive valuations). It is the greedy strategy that in every one of the agent's picking rounds picks the most desirable item among those remaining (breaking ties arbitrarily). The value guaranteed to the agent $i$ who uses the greedy strategy in a picking sequence in which $R_i = \{r_1, r_2, \ldots \}$ is the set of her picking rounds is that of the bundle that contains those items of indices $\{r_1, r_2, \ldots \}$, if items are sorted in order of decreasing desirability (increasing disutility) for the agent.

In this section we consider potential sources of envy among agents, and ways of addressing, or at least partly addressing, such concerns. 

We first classify types of envy that we will be considering. We shall consider envy of an agent $p$ towards an agent $q$. Our classification involves two attributes. The first attribute refers to the entitlement, stating whether $p$ and $q$ have equal entitlement, and if not, which of the two has higher entitlement. The second attribute refers to the timing of the envy. Ex-post envy means that after the allocation is complete, agent $p$ prefers the bundle $B_q$ received by $q$ over her own bundle $B_p$. Ex-ante envy refers to envy before the allocation mechanism is actually run. As it involves the beliefs of the agent concerning what would happen when the mechanism is run, we need to model this belief. In this paper we assume that the agent is risk averse, aiming to maximize her guaranteed utility. (This assumption is made when the allocation mechanism is deterministic. We shall later extend this definition to randomized allocation mechanisms.) For picking sequences, this means that agent $p$ envies agent $q$ if the value (under valuation function $v_p$) guaranteed by the greedy strategy on the set $R_q$ of picking rounds is strictly better than the value guaranteed on the set $R_p$. 

Sometimes, when given a picking sequence $\pi$, we shall consider allocation mechanisms that include a preliminary phase in which it is determined for each agent which label of the picking sequence she gets. In this case, the allocation mechanism includes two phases (the preliminary one, and then the item picking phase). Ex-ante considerations refer to the stage before the first phase. Likewise, risk aversion refers to the complete strategy over both phases.

\subsection{Envy for agents with equal entitlement}

{\bf Ex-post envy is unavoidable.}
There are allocation instances with equal entitlement that do not have any envy-free allocation that allocates all items (e.g., an instance with only one item). Hence we shall not be concerned with ex post envy for our picking sequences (in settings with equal entitlement). 

{\bf Ex-ante envy is avoidable.} 
Any picking sequence $\pi$ can be made ex-ante envy free by adding a preliminary phase in which one picks a random permutation over the names of the agents, so that each agent has probability $\frac{1}{n}$ of getting each of the $n$ labels associated with the picking sequence. Using this preliminary phase, risk averse agents do not envy other agents ex-ante, and also receive in expectation a bundle of value not worse than their proportional share (regardless of the strategies of other agents). 

{\bf A picking sequence for labels.}
In a picking sequence $\pi$ there are $n$ labels, each associated with a single agent.
Here we propose a preliminary phase that we refer to as {\em random picking sequence for labels}. First select a uniformly random permutation $\sigma$ over the agents. 
Thereafter, each agent in her turn (according to $\sigma$) selects a label among those labels that are still available. This completes the preliminary phase. Thereafter, in the item picking phase, the picking sequence $\pi$ is {used}, 
where every agent selects items in the rounds associated with her label {in $\pi$}. 

Let us present an example for the use of a picking sequence for labels. Suppose that there are two agents ($A$ and $B$), three indivisible chores ($e_1, e_2, e_3$), and the picking sequence (for items) $(1,1,2)$ (the agent labeled one first picks two chores, and the agent labeled~2 gets the remaining chore). Suppose that the additive valuations (disutilities) over the chores are $(6,4,4)$ for agent $A$ and $(6,2,2)$ for agent $B$. Then whichever agent gets label~1, that agent will pick items $e_2$ and $e_3$, and the other agent will get item $e_1$.

If we just use $\pi$ arbitrarily, we might give agent $A$ label~1, and then agent $A$ suffers a disutility of~8, whereas agent $B$ suffers a disutility of~6.

If we decide at random which agent gets label~1 (eliminating ex-ante envy), then in expectation agent $A$ gets a disutility of~7, whereas agent $B$ gets a disutility of~5. Hence the ex-ante disutility levels decrease (to the proportional share), whereas the ex-post guarantees are not harmed.

If we use a random picking sequence for labels, then a risk averse agent $A$ will pick label~2 and a risk averse agent will $B$ pick label~1. (The two agents have the same ordinal preference over items -- the instance is IDO. However, they have different cardinal preferences, and this causes them to prefer different labels.) Then, using $\pi$, agent $A$ gets a disutility of~6, whereas agent $B$ gets a disutility of~4. Hence the ex-ante disutility is even smaller than the proportional share, and likewise for the ex-post utility.

\begin{proposition}
\label{pro:equal}
Let $\pi$ be a picking sequence for allocation of $m$ indivisible chores to $n$ agents who have additive disvaluation functions. If we employ a preliminary phase of a random picking sequence for labels, then the resulting mechanism has the following properties:

\begin{enumerate}
    \item Every risk averse agent has no ex-ante envy towards any other agent.
    \item The expected disvalue of the bundle received by a risk averse agent is no worse than her proportional share, regardless of the strategies of all other agents.
    \item There are IDO instances such that ex-post every agent gets a bundle of value strictly better than her proportional share.
    
\end{enumerate}
\end{proposition}

\begin{proof}
We address below the three properties.

\begin{enumerate}
    \item Ex-ante (before the preliminary phase) the roles of all agents are symmetric, and hence a risk averse agent has no ex-ante envy. (An agent $i$ that is not risk averse according to our definitions might have envy. For example, $i$ might believe that some other agent $j$ has the same valuation function as $i$ does, and that all other agents intend to use strategies that give $j$ the best bundle possible. An agent $i$ with such beliefs envies $j$, in the sense that she would like to switch identities with $j$. But a risk averse $i$ would not care about switching identities with $j$, because she believes that regardless of her identity, the strategies of other agents are those that will minimize her utility.)
    
    \item Fix a risk averse agent $i$. For each label $j$, let $v_i^j$ denote the value that the greedy picking strategy guarantees to agent $i$ if in the picking order $\pi$ agent $i$ has label $j$. Then the expected value of $v_i^j$ (over a uniformly random choice of $j$) equals the proportional share of $i$. Order the labels in order of decreasing order of $v_i^j$ (so that $v_i^1 \ge v_i^2 \ge \ldots \ge v_i^n$). The risk averse agent $i$ employs a greedy picking strategy in the preliminary phase of random picking sequence for labels, meaning that when it is her turn to pick a label, she picks the best among the remaining labels. Under this strategy, the distribution over labels that agent $i$ obtains stochastically dominates the uniform distribution (for every $j$, she has probability at least $\frac{j}{n}$ of getting one of her $j$ most preferred labels). Consequently, in expectation, the value of the bundle that $i$ receives is no worse than her proportional share.
    
    \item We have seen such an example with two agents before the proposition. The example generalizes to any number $n$ of agents as follows. There are $n+1$ chores. Each of the first $n-1$ chores has disutlity~3 for every agent.  For $n-1$ agents, the last two chores each has disutility~2, whereas for one agent, they each have disultility~1. In the picking order, the first labeled agent picks two items, and then each of the remaining agents picks one item. (A similar example can be designed for goods.) 
    \end{enumerate}
\end{proof}

Corollary~\ref{cor:equal} is proved by combining Theorem~\ref{thm:equal} with Proposition~\ref{pro:equal}.

\subsection{Envy for agents with unequal entitlement}

For allocation of goods, we expect agents with higher entitlement to get bundles of higher value than agents with lower entitlements. Consequently, it is natural that agents of low entitlement envy the agents of high entitlement, and we do not attempt to eliminate such envy. However, it seems inappropriate that in a given allocation, agents of high entitlement envy agents of lower entitlement. Hence this is the type of envy that we would like to avoid. Likewise, for allocation of chores, we wish there to be no envy of agents of low responsibility towards agents of high responsibility. 

For picking sequences over chores for agents with arbitrary responsibility (and additive disvaluations), there is a necessary and sufficient condition that ensures that there is no envy of one agent towards the other. (An analogous condition holds for the case of goods but is omitted here.)

\begin{proposition}
\label{pro:EnvyCondition}
Let $\pi$ be a picking sequence for chores. Then a risk averse picker $i$ (with an additive disvaluation function) does not envy picker $j$ (not even ex-post) if in every suffix of $\pi$, picker $j$ has at least as many picks as agent $i$. (For goods the condition is that in every prefix of $\pi$, picker $j$ has at least as many picks as agent $i$.) Moreover, if the condition fails to hold, then there is a choice of additive disvaluation function for picker $i$ under which $i$ envies $j$.
\end{proposition}

\begin{proof}
We only explain why the condition is necessary. Suppose that in a suffix of length $\ell$ agent $i$ has $k \ge 1$ picks whereas agent $j$ has fewer picks (without loss of generality, $k - 1$ picks). Then if the additive disvaluation function of $i$ is such that $m - \ell$ items have disvalue~0 and each of the remaining $\ell$ items has disvalue~1, picker $i$ envies picker $j$. 
\end{proof}

For an arbitrary vector of responsibilities, it is not difficult to design picking sequences in which no risk averse agent envies an agent of higher responsibility, and moreover, every risk averse agent gets a bundle of disvalue at most twice her APS. Start with the proportional fractional allocation, and round it to a picking sequence as in Lemma~\ref{lem:FracToSequence} (this ensures a factor two approximation to the APS). In this rounding (as described in the proof of the lemma), in every round $t$, if there is a choice of several eligible pickers, use the one of highest responsibility. This will satisfy the conditions of Proposition~\ref{pro:EnvyCondition}, ensuring no envy (not even ex-post) of agents of low responsibility towards agents of high responsibilities. 

However, if we wish to have approximation factors $\rho < 2$ (for a fixed $\rho$ independent of $n$)  compared to the APS, this conflicts with ex-post envy freeness.

\begin{proposition}
Consider picking sequences for allocating indivisible chores to agents with additive disvaluation functions. For every $n$, there is a vector of responsibilities for the agents, such that for every picking sequence, either an agent with lower responsibility might envy an agent with higher responsibility, or an agent might get a bundle of value $(2 - O(\frac{1}{n}))$ times her APS. 
\end{proposition}

\begin{proof}
Consider $n \ge 3$ agents, $m = kn+1$ chores (where $k = n-2$), and the vector $(\frac{k+1}{m}, \frac{k}{m}, \ldots, \frac{k}{m})$ of responsibilities. We show that in every picking sequence, either an agent with responsibility $\frac{k}{m}$ might envy the agent with responsibility $\frac{k+1}{m}$, or the agent of responsibility $\frac{k+1}{m}$ might get a bundle of disvalue $2 - \frac{4}{n}$ times her APS.

We refer to the agent with responsibility $\frac{k+1}{m}$ as agent~1. If no agent with lower responsibility envies agent~1, then by Proposition~\ref{pro:EnvyCondition}, in the picking sequence agent~1 picks in round $m$, and also at least $k$ times in the other rounds. Consider a disvaluation function $v_1$ for agent~1 in which chore $e_1$ has disvalue~$k$, and each of the $kn$ other chores has disvalue~1. In the picking sequence, agent~1 might get a bundle of disvalue $2k$, whereas her APS is at most $k+2$. (In any pricing function in which chore prices sum up to $kn + 1$, either $e_1$ is priced at least $k+1$, or the $k+2$ other chores of highest price have total price at least $(k+2)\frac{kn+1 - (k + 1)}{kn} = k+1$, where equality holds because $k = n-2$.) 
\end{proof}

As we do want to get approximation ratios better than~2 compared to the APS, we shall allow for ex-post envy of an agent towards an agent of higher responsibility. However, given any picking sequence $\pi$, such ex-ante envy can be eliminated by the use of an auxiliary picking sequence $\sigma$ for picking identities. Specifically, $\sigma$ orders the agents from lowest responsibility to highest, breaking ties uniformly at random. Risk averse agents will not have ex-ante envy towards agents with higher or equal entitlement. This together with Theorem~\ref{thm:arbitrary} proves Corollary~\ref{cor:arbitrary}. 

\section{Analysis of AlgChores}


{
In this section we consider allocation of indivisible chores to agents of equal entitlement, and 
prove that the polynomial time algorithm {\em AlgChores}, introduced in~\cite{BarmanK20}, gives every agent a bundle of disvalue not larger than $\frac{4n-1}{3n}$ times her APS, 
thus proving Theorem~\ref{thm:APS}.
Specifically, Lemma~\ref{lem:APS} proves Theorem~\ref{thm:APS} for the special case of IDO instances (see Definition~\ref{def:IDO}).  As explained in~\cite{BarmanK20}, this suffices in order to imply that the version of AlgChores that is run on general instances also has the same approximation ratio as proved in Lemma~\ref{lem:APS}, thus establishing Theorem~\ref{thm:APS}.

\begin{algorithm}
\KwIn{An IDO instance}
\KwOut{An allocation}
\BlankLine
	Initialize allocation $B$ to empty\;
	\For{$r\rightarrow 1$ to $m$}
	{
	Let agent $i$ be an agent that envies no other agent\;
	Allocate chore $e_r$ to bundle $B_i$\;
	Resolve envy cycles\;
	}
    
	\Return {Allocation $B$}
\caption{AlgChores}\label{alg-equalaps}	
\end{algorithm}

\begin{lemma}
\label{lem:APS}
In every IDO instance in which $n$ agents have additive valuations over chores and equal entitlements, the AlgChores allocation gives every agent $i$ a bundle of disvalue at most $\frac{4n-1}{3n}APS_i$
\end{lemma}

\begin{proof}
Recall that in IDO instances, items are ordered in decreasing order of disvalue, with $e_1$ being the item of highest disvalue. Recall that in AlgChores, item $e_r$ is allocated in round $r$ to an agent that envies no other agent, and then the round is completed by eliminating envy-cycles (a more detailed description can be found in~\cite{BarmanK20}). For simplicity of notation, we assume without loss of generality that valuations are scaled so that the APS of every agent is~1. As the APS is not smaller than the PS (proportional share), in every round $r$ AlgChore allocates item $e_r$ to an agent $i$ who holds a bundle of disvalue at most her APS. (As $i$ did not envy any other agent at the beginning of round $r$, her disvalue was not larger than her proportional share, and hence not larger than her APS.)

Consider an arbitrary agent $a$ whose final bundle has disvalue strictly larger than~1 (if there is no such agent then we are done). Let $t$ be smallest so that for all rounds $t$ and above the bundles that $a$ holds have disvalue larger than~1. (Agent $a$ does not get any new chore after round $t$, because there always is at least one bundle of value less than the proportional share, which is at most 1.)  Then in the end of round $t-1$ (after envy-cycle elimination, if round $t-1$ created an envy-cycle) agent $a$ holds a bundle $B^{t-1}$ of disvalue at most~1. 
(In fact, as will be shown later in the proof, the fact that agent $a$ gets item $e_t$ implies that $c_a(B^{t-1}) \le 1 - \frac{c_a(e_t)}{n}$.) Let $e_i$ be the item of highest disvalue in $B^{t-1}$. ($B^{t-1}$ cannot be empty because no single item has disvalue more than the APS, and consequently $e_t$ would not bring the disvalue of the bundle to above the APS). Rename agent $a$ to be agent $i$ (and $B^{t-1}$ to be $B_i^{t-1}$), and likewise, rename all other agents according to the item of highest disvalue in their bundle at the end of round $t-1$. (Note that it could be that the first item that agent $a$ received was not $e_i$, if the agent subsequently participated in the elimination of an envy cycle.) With a slight abuse of notation, we denote  $B_i^{t-1} \cup \{e_t\}$ by $B_i^t$ {(this is an abuse of notation, because if round $t$ ends with an envy cycle elimination, then the bundle held by agent $i$ after round $t$ might not be $B_i^{t-1} \cup \{e_t\}$)}. The disvalue of the final bundle received by agent $i$ is not larger than $c_i(B_i^t)$, as after round $t$ agent $i$ never receives an item (though $i$ can participate in elimination of envy cycles). Hence it remains to bound $c_i(B_i^t)$.

Consider a price function in which for every $1 \le j \le n$ (including $j = i$), every item in bin $B_j^{t-1}$ gets price $\frac{1 - \epsilon}{|B_j^{t-1}|}$, item $e_t$ gets price $n\epsilon$, and the remaining items (of index higher than $t$) have price~0. Here $\epsilon > 0$ is chosen to be sufficiently small so that $e_t$ has the smallest price among all items of positive price. By definition of the APS, there is a bundle $B$ of price at least~1 whose disvalue is at most~1. Being of price~1, bundle $B$ must contain at least two items. Suppose that $B$ contains at least three items. As the instance is ordered and $c_i(B) \le 1$, this implies that $c_i(e_t)\le \frac{1}{3}$. As we are running AlgChores, we also have the inequalities $c_i(B_i^{t-1}) \le c_i(B_j^{t-1})$ for every $1 \le j \le n$. As the APS is at least as large as the proportional share we have the inequality $n \cdot APS_i = n \ge c_i(\items) \ge c_i(e_t) + \sum_{j=1}^n c_i(B_j^{t-1}) \ge c_i(e_t) + n\cdot c_i(B_i^{t-1})$. Hence $c_i(B_i^{t-1}) \le 1 - \frac{c_i(e_t)}{n}$. This implies that $c_i(B_i^t) = c_i(B^{t-1}) + c_i(e_t) \le 1 + (1 - \frac{1}{n})c_i(e_t) \le \frac{4n-1}{3n}$. 
  
It remains to prove that $B$ does not contain exactly two items. Suppose for the sake of contradiction that $B$ contains two items. {Then one of the items of $B$ (say $e_k$) must be priced $1 - \epsilon$ (because any item of smaller price has price at most $\frac{1-\epsilon}{2}$).} By the pricing rule, the corresponding $B_k^{t-1}$ has a single item. Moreover, $c_i(\{e_k , e_t\}) \le c_i(B) \le 1$, where the first inequality follows from the fact that $B$ contains an item of positive cost other than $e_k$, and among items of positive cost, $e_t$ has the smallest disvalue. Consequently, $c_i(B_k^{t-1}) < c_i(B_i^{t-1})$ (because $c_i(B_i^{t-1}) + c_i(e_t) = c_i(B_i^t) > 1$), implying that agent $i$ envied bin $B_k^{t-1}$ in round $t$. This contradicts the fact that $e_t$ is given to agent $i$.  
\end{proof}

AlgChores does not guarantee a ratio better than $\frac{4n - 1}{3n}$ (not even compared to the MMS). Consider $n$ agents, all having the same additive disvaluation function $c$. There are $2n+1$ items, indexed both by subscripts and superscripts (items with the same subscript have the same disvalue). Three items $e_0^1$, $e_0^2$, $e_0^3$ each has disvalue~1. 
For each $1 \le j \le n-1$, there are two items $e_j^1$ and $e_j^2$, each of disvalue $1 + \frac{j}{n}$. The MMS is~3 (the MMS partition has one bundle with $\{e_0^1, e_0^2, e_0^3\}$, and for every $1 \le j \le n-1$, the bundle $\{e_j^1, e_{n-j}^2\}$). AlgChores will allocate chores starting with the chore of highest disvalue. After allocating $2n$ chores (all chores except for chore $e_0^1$), it will have $n$ bundles, each of value $3n - \frac{1}{n}$. Thereafter, regardless of where $e_0^1$ is placed, it will creates a bundle of cost $4 - \frac{1}{n}$. Compared to the MMS (and likewise to the APS, which equals the MMS in this allocation instance), this gives a ratio of $\frac{4n - 1}{3n}$. 
}

The example shown in Section~\ref{sec:CS} implies that there are instances in which in every allocation some agent $i$ gets a bundle of at least $\frac{3n}{2n-1}CS_i$. Consequently, $APS_i$ cannot be replaced by $CS_i$ in Lemma~\ref{lem:APS}.

\subsection*{Acknowledgements}

This research was supported in part by the Israel Science Foundation (grant No. 1122/22). Xin Huang was supported in part at the Technion Israel Institute of Technology by an Aly Kaufman Fellowship. 

\bibliographystyle{acm}
\bibliography{ref}

\begin{thebibliography}{10}

\bibitem{amanatidis2015approximation}
{\sc Amanatidis, G., Markakis, E., Nikzad, A., and Saberi, A.}
\newblock Approximation algorithms for computing maximin share allocations.
\newblock In {\em International Colloquium on Automata, Languages, and
  Programming\/} (2015), pp.~39--51.

\bibitem{aziz2018fair}
{\sc Aziz, H., Caragiannis, I., Igarashi, A., and Walsh, T.}
\newblock Fair allocation of indivisible goods and chores.
\newblock In {\em Proceedings of the Twenty-Eighth International Joint
  Conference on Artificial Intelligence, {IJCAI} 2019, Macao, China, August
  10-16, 2019\/} (2019), S.~Kraus, Ed., ijcai.org, pp.~53--59.

\bibitem{aziz2022algorithmic}
{\sc Aziz, H., Li, B., Moulin, H., and Wu, X.}
\newblock Algorithmic fair allocation of indivisible items: A survey and new
  questions.
\newblock {\em arXiv preprint arXiv:2202.08713\/} (2022).

\bibitem{aziz2022approximate}
{\sc Aziz, H., Li, B., and Wu, X.}
\newblock Approximate and strategyproof maximin share allocation of chores with
  ordinal preferences.
\newblock {\em Mathematical Programming\/} (2022), 1--27.

\bibitem{aziz2017algorithms}
{\sc Aziz, H., Rauchecker, G., Schryen, G., and Walsh, T.}
\newblock Algorithms for max-min share fair allocation of indivisible chores.
\newblock In {\em Thirty-First AAAI Conference on Artificial Intelligence\/}
  (2017).

\bibitem{BEFBoBW}
{\sc Babaioff, M., Ezra, T., and Feige, U.}
\newblock Best-of-both-worlds fair-share allocations.
\newblock {\em CoRR abs/2102.04909\/} (2021).

\bibitem{babaioff2021fair}
{\sc Babaioff, M., Ezra, T., and Feige, U.}
\newblock Fair-share allocations for agents with arbitrary entitlements.
\newblock In {\em Proceedings of the 22nd ACM Conference on Economics and
  Computation\/} (2021), pp.~127--127.

\bibitem{BarmanK20}
{\sc Barman, S., and Krishnamurthy, S.~K.}
\newblock Approximation algorithms for maximin fair division.
\newblock {\em {ACM} Trans. Economics and Comput. 8}, 1 (2020), 5:1--5:28.

\bibitem{barman2018finding}
{\sc Barman, S., Krishnamurthy, S.~K., and Vaish, R.}
\newblock Finding fair and efficient allocations.
\newblock In {\em Proceedings of the 19th ACM Conference on Economics and
  Computation (EC)\/} (2018), pp.~557--574.

\bibitem{berger2022almost}
{\sc Berger, B., Cohen, A., Feldman, M., and Fiat, A.}
\newblock Almost full {EFX} exists for four agents.
\newblock In {\em Thirty-Sixth {AAAI} Conference on Artificial Intelligence,
  {AAAI} 2022\/} (2022), {AAAI} Press, pp.~4826--4833.

\bibitem{BM01}
{\sc Bogomolnaia, A., and Moulin, H.}
\newblock A new solution to the random assignment problem.
\newblock {\em J. Econ. Theory 100}, 2 (2001), 295--328.

\bibitem{bouveret2016characterizing}
{\sc Bouveret, S., and Lema{\^\i}tre, M.}
\newblock Characterizing conflicts in fair division of indivisible goods using
  a scale of criteria.
\newblock {\em Autonomous Agents and Multi-Agent Systems 30}, 2 (2016),
  259--290.

\bibitem{budish2011combinatorial}
{\sc Budish, E.}
\newblock The combinatorial assignment problem: Approximate competitive
  equilibrium from equal incomes.
\newblock {\em Journal of Political Economy 119}, 6 (2011), 1061--1103.

\bibitem{caragiannis2016unreasonable}
{\sc Caragiannis, I., Kurokawa, D., Moulin, H., Procaccia, A.~D., Shah, N., and
  Wang, J.}
\newblock The unreasonable fairness of maximum nash welfare.
\newblock In {\em Proceedings of the 17th ACM Conference on Economics and
  Computation (EC)\/} (2016), pp.~305--322.

\bibitem{chaudhury2022fair}
{\sc Chaudhury, B.~R., Cheung, Y.~K., Garg, J., Garg, N., Hoefer, M., and
  Mehlhorn, K.}
\newblock Fair division of indivisible goods for a class of concave valuations.
\newblock {\em Journal of Artificial Intelligence Research 74\/} (2022),
  111--142.

\bibitem{DBLP:conf/sigecom/ChaudhuryGM20}
{\sc Chaudhury, B.~R., Garg, J., and Mehlhorn, K.}
\newblock {EFX} exists for three agents.
\newblock In {\em {EC} '20: The 21st {ACM} Conference on Economics and
  Computation, Virtual Event, Hungary, July 13-17, 2020\/} (2020),
  P.~Bir{\'{o}}, J.~Hartline, M.~Ostrovsky, and A.~D. Procaccia, Eds., {ACM},
  pp.~1--19.

\bibitem{chaudhury2021improving}
{\sc Chaudhury, B.~R., Garg, J., Mehlhorn, K., Mehta, R., and Misra, P.}
\newblock Improving efx guarantees through rainbow cycle number.
\newblock In {\em Proceedings of the 22nd ACM Conference on Economics and
  Computation\/} (2021), pp.~310--311.

\bibitem{feige2021tight}
{\sc Feige, U., Sapir, A., and Tauber, L.}
\newblock A tight negative example for mms fair allocations.
\newblock In {\em International Conference on Web and Internet Economics\/}
  (2021), Springer, pp.~355--372.

\bibitem{GargMT19}
{\sc Garg, J., McGlaughlin, P., and Taki, S.}
\newblock Approximating maximin share allocations.
\newblock In {\em 2nd Symposium on Simplicity in Algorithms, SOSA@SODA 2019,
  January 8-9, 2019 - San Diego, CA, {USA}\/} (2019), J.~T. Fineman and
  M.~Mitzenmacher, Eds., vol.~69 of {\em {OASICS}}, Schloss Dagstuhl -
  Leibniz-Zentrum f{\"{u}}r Informatik, pp.~20:1--20:11.

\bibitem{garg2019improved}
{\sc Garg, J., and Taki, S.}
\newblock An improved approximation algorithm for maximin shares.
\newblock In {\em {EC} '20: The 21st {ACM} Conference on Economics and
  Computation, Virtual Event, Hungary, July 13-17, 2020\/} (2020),
  P.~Bir{\'{o}}, J.~Hartline, M.~Ostrovsky, and A.~D. Procaccia, Eds., {ACM},
  pp.~379--380.

\bibitem{ghodsi2018fair}
{\sc Ghodsi, M., Hajiaghayi, M., Seddighin, M., Seddighin, S., and Yami, H.}
\newblock Fair allocation of indivisible goods: Improvements and
  generalizations.
\newblock In {\em Proceedings of the 19th ACM Conference on Economics and
  Computation (EC)\/} (2018), pp.~539--556.

\bibitem{ghodsi2022fair}
{\sc Ghodsi, M., Hajiaghayi, M., Seddighin, M., Seddighin, S., and Yami, H.}
\newblock Fair allocation of indivisible goods: Beyond additive valuations.
\newblock {\em Artificial Intelligence 303\/} (2022), 103633.

\bibitem{huang2021algorithmic}
{\sc Huang, X., and Lu, P.}
\newblock An algorithmic framework for approximating maximin share allocation
  of chores.
\newblock In {\em Proceedings of the 22nd ACM Conference on Economics and
  Computation\/} (2021), pp.~630--631.

\bibitem{procaccia2014fair}
{\sc Kurokawa, D., Procaccia, A.~D., and Wang, J.}
\newblock Fair enough: Guaranteeing approximate maximin shares.
\newblock {\em Journal of the ACM 65(2)\/} (2018), 8:1--27.

\bibitem{li2022fair}
{\sc Li, B., Wang, F., and Zhou, Y.}
\newblock Fair allocation of indivisible chores: Beyond additive valuations.
\newblock {\em arXiv preprint arXiv:2205.10520\/} (2022).

\bibitem{livanos2022almost}
{\sc Livanos, V., Mehta, R., and Murhekar, A.}
\newblock ({Almost}) {Envy}-free, proportional and efficient allocations of an
  indivisible mixed manna.
\newblock In {\em Proceedings of the 21st International Conference on
  Autonomous Agents and Multiagent Systems\/} (2022), pp.~1678--1680.

\end{thebibliography}

\end{document}